\documentclass[11pt,a4paper]{article}
\usepackage[T1]{fontenc}
\usepackage{lmodern}
\usepackage{amsmath,amssymb,amsthm}
\usepackage[
  a4paper,
  left=20mm,
  right=20mm,
  top=18mm,
  bottom=18mm,
  headheight=14pt
]{geometry}
\usepackage{microtype}
\usepackage{enumitem}
\usepackage{booktabs,longtable,array}
\usepackage[unicode,hidelinks]{hyperref}
\usepackage[title,titletoc]{appendix}
\hypersetup{
  pdftitle={Packing Tails of Reciprocal Rectangles into Squares of Equal Area},
  pdfsubject={Detailed proof of an equal-area packing theorem for sufficiently late tails},
  pdfkeywords={Meir--Moser packing problem; equal-area packing;
probabilistic method.}
}
\newtheorem{theorem}{Theorem}[section]
\newtheorem{lemma}[theorem]{Lemma}
\newtheorem{corollary}[theorem]{Corollary}
\numberwithin{equation}{section}
\setlist{topsep=5pt,itemsep=3pt,parsep=0pt,leftmargin=*}
\title{Packing Tails of Reciprocal Rectangles\\into Squares of Equal Area}
\author{Yu Jiang\\\texttt{yuj@utexas.edu}}
\date{\today}
\begin{document}
\maketitle

\begin{abstract}
The Meir--Moser rectangle-packing problem asks whether all rectangles with side lengths \(1/n\) and \(1/(n+1)\), for \(n\ge1\), can be packed into the unit square with pairwise disjoint interiors. We establish a tail version of this problem. Let \(R_n\) denote the rectangle with these side lengths. We prove that there exists an integer \(m_0\) such that, for every \(m\ge m_0\), the family \(\{R_n:n\ge m\}\) admits a packing, by translations and right-angle rotations, into a square of side length \(m^{-1/2}\), with pairwise disjoint interiors. The area of the square equals the sum of the areas of all the rectangles.
The geometric construction recursively decomposes rectangular gaps, while local randomized quotas and random permutations assign subsequent integer indices. We separately control the total area of waiting gaps and the assignment load at each index. 
The proof is organized in six steps: a finite-prefix reduction, geometric row decompositions, an area bootstrap, a sharp source-load estimate, control of the actual adaptive construction, and a compactness limit. For every finite time horizon, the probability of failure has a bound that is independent of the horizon and can be made arbitrarily small. The adaptive step uses a permanent load ledger, actual fresh height queries, and a one-sided comparison with a frozen source experiment. Compactness then yields an infinite packing. 
The final packing statements have been checked in Lean 4. 
A sufficient threshold is \(m_0=10^{1000}\). This result applies only to sufficiently late tails and does not resolve the original Meir--Moser rectangle-packing problem for the full sequence starting at \(n=1\), which remains open.

\end{abstract}


\section*{Introduction}

The Meir--Moser rectangle-packing problem asks whether all rectangles of dimensions \(1/n\times1/(n+1)\), \(n\ge1\), can be packed into the unit square with pairwise disjoint interiors \cite{meir1968packing}.
Several related advances motivate the present work. Tao \cite{tao2024perfectly} proved that, for \(1/2<t<1\), sufficiently late tails of squares of side length \(n^{-t}\) perfectly pack a square of equal total area, and noted that the method also applies to rectangles of dimensions \(n^{-t}\times(n+1)^{-t}\).
Zhu and Jo\'os \cite{zhu2022packing} constructed a packing of the first \(1.35\times10^{11}\) reciprocal rectangles into the unit square.
Kislovskiy et al.~\cite{kislovskiy2026slack} introduced the Slack-Pack algorithm, obtaining conditional estimates and numerical evidence for perfect tail packings at the harmonic endpoint \(t=1\).
The present paper establishes an equal-area packing theorem for sufficiently late tails of the rectangles \(1/n\times1/(n+1)\). 
The original problem for the full sequence starting at \(n=1\)
remains open.

\clearpage

\tableofcontents

\clearpage

\section{Main theorem, conventions, and parameters}\label{sec:1}

\begin{theorem}[Tail packing]
\label{thm:1.1}
There exists an integer \(m_0\ge1\) such that, for every integer \(m\ge m_0\), one can choose, for each \(n\ge m\), an axis-parallel rectangle \(\widetilde R_n\) congruent to

\[
R_n=[0,1/n]\times[0,1/(n+1)]
\]

such that

\[
\widetilde R_n\subseteq Q_m:=[0,m^{-1/2}]^2,
\qquad
\operatorname{int}\widetilde R_i\cap
\operatorname{int}\widetilde R_j=\varnothing\quad(i\ne j).
\tag{1.1}\label{eq:1.1}
\]

Moreover,

\[
\left|Q_m\setminus\bigcup_{n\ge m}\widetilde R_n\right|=0.
\tag{1.2}\label{eq:1.2}
\]

\end{theorem}

Throughout the paper, the orientations of the two sides of a rectangle may be interchanged. In local coordinates, the longer side is called its width $W$, and the shorter side its height $H$; these local rotations preserve global axis parallelism.

\subsection{Proof overview: six steps and their purpose}\label{sec:overview}

The difficulty is to keep enough usable space and enough unused indices available at the same time. An area estimate alone does not say that the remaining pieces fit; an allocation estimate alone does not say that the remaining regions have the right shapes. We keep these two requirements separate until the final existence argument.

\begin{enumerate}[label=\textbf{Step \arabic*.}]
\item \textbf{Reduce the goal to finite prefixes.} Fix \(m\) and a finite last index \(L\). It will be enough to pack \(R_m,\ldots,R_L\) into the \textbf{same square} for every \(L\). The finite packings need not extend one another. Closedness and compactness will recover a single infinite packing in Section~\ref{sec:13}.
\item \textbf{Replace each used row by controlled smaller gaps.} A row of height \(H\) uses indices near \(H^{-1}\). Each tile leaves a much thinner principal gap; the endpoint is split into two new rectangular jobs. Sections~\ref{sec:2}--\ref{sec:3} show that these operations preserve a common shape condition and that the endpoint trees have a finite total boundary budget.
\item \textbf{Keep the primary rectangle large enough.} Remove a strip from the primary rectangle whenever an unassigned index needs a main row. Before processing \(N\), the remaining area is exactly \(1/N\). Sections~\ref{sec:5}--\ref{sec:7} bound the waiting gaps and the already reserved future tiles, leaving a fixed positive fraction of \(1/N\) in the primary rectangle. Its aspect ratio stays at most two, so a new strip fits.
\item \textbf{Show that future index demand stays below capacity.} Many inherited jobs can ask for indices in the same range. Section~\ref{sec:4} spreads their integer demands using one random rounding mark per job and one permutation per bin. For a target \(N\), its principal gaps come from a smaller scale \(S=N^{4/5}\). Section~\ref{sec:9} assigns one unit of comparison mass to each source position and obtains the leading mean coefficient \(1/p=4/5\), below the permitted principal load \(0.85\).
\item \textbf{Apply the probability estimates to the actual construction.} The source pool is itself random, and an earlier failure may stop the construction. Sections~\ref{sec:11}--\ref{sec:12} freeze the pool before its relevant random marks are read, then compare its output with the actual running ledger. The key inequality is
\[
 F_P(N,t)\le Z_N+Q_t-W_t.
\]
Here \(Z_N\) is a frozen source cost, \(Q_t\) is the sum of actual main-row load queries, and \(W_t\) is their used comparison weight. Except on a small source event, \(Z_N<0.83\). A principal overload therefore forces \(Q_t-W_t\ge0.02\) at all sufficiently late program times. A uniform probability bound for this persistent deviation avoids a factor depending on the number of queries.
\item \textbf{Choose a successful finite packing, then pass to the limit.} Endpoint and bin errors have summable tails, and the least principal obstruction has probability at most \(4e^{-\kappa N^{7/50}}\), with \(\kappa>0\). The union of all obstructions has probability less than one for large \(m\), uniformly in \(L\). Thus every finite prefix has a packing. Compactness supplies the infinite family, and the telescoping area sum proves that its uncovered set has measure zero.
\end{enumerate}

The logical dependencies are therefore
\[
\begin{gathered}
\text{row geometry}\ \Longrightarrow\ \text{area and shape invariants},\\
\text{integer allocation}\ \Longrightarrow\ \text{source and query tail bounds},\\
\text{both, with the actual revelation order}\ \Longrightarrow\
\text{finite-prefix existence}\ \Longrightarrow\ \text{infinite packing}.
\end{gathered}
\]

Four sections establish the main transitions in the proof.
Section~\ref{sec:7} shows how the area estimates guaranty that the primary rectangle retains enough space for each new main row.
Section~\ref{sec:9.4} identifies the mass partition that yields the leading mean-load coefficient \(1/p=4/5\), leaving a positive margin below the principal load allowed.
Section~\ref{sec:12.2} connects the frozen source experiment to the actual adaptive construction through a step-by-step comparison of their load ledgers.
Section~\ref{sec:13} passes from finite-prefix packings to an infinite packing and proves the equal-area conclusion.
The remaining sections develop the ingredients needed for these transitions: geometric decompositions and shape bounds, integer allocation rules, area and probability estimates, and the revelation order that justifies their use in the adaptive construction.

\subsection{Objects, two ledgers, and indexing}\label{sec:dictionary}

A \emph{job} is a rectangular gap together with a future set of eligible tile indices and an integer demand. A \emph{main row} is cut from the primary rectangle; an \emph{inherited row} fills a previously created principal or endpoint gap. Thus, ``main row'' and ``principal gap'' name different objects.

The \emph{geometric ledger} records the current disjoint regions: the primary rectangle, waiting gaps, and placed or reserved tile positions. Once a gap is filled, it leaves the waiting list. The \emph{permanent demand ledger} retains every job after registration, even after it starts. Its contribution at index \(N\) is \(k/L(H)\) when \(N\in I(H)\), and zero otherwise. This second ledger makes \(F_P(N,t)\) and \(F_E(N,t)\) non-decreasing in \(t\). It deliberately over-counts demand from old jobs, which is safe for an upper bound. It must never be used as a list of currently empty geometric regions.

All tile indices are positive integers. A finite prefix through \(L\) includes \(R_L\). A cursor at \(N\) means that indices smaller than \(N\) have already been placed. An empty source pool, zero quota, empty family of queried groups, or empty endpoint forest is allowed.

\begin{table}[tbp]
\centering
\small
\begin{tabular}{@{}p{.18\textwidth}p{.75\textwidth}@{}}
\toprule
Symbol & Meaning \\
\midrule
\(m,L,N\) & Initial tile index, finite terminal tile index, and target index at which demand is tested. \\
\(S=N^{4/5}\) & Source scale for principal gaps that can ask for target \(N\). \\
\(t\) & Discrete program time, including deterministic transitions and random requests; distinct from the physical cursor, which tracks tile indices. \\
\(I(H),L(H)\) & Eligible positive integer indices of a job of height \(H\), and their number. \\
\(\mathcal B_b,s_b,\ell_b\) & A bin, its first index, and its length. \\
\(F_P(N,t),F_E(N,t)\) & Fractional principal and endpoint demand at \(N\), summed over all registered jobs at program time \(t\). \\
\(D_P(h)\) & An area bound for principal gaps of height at most \(h\); it is not an allocation load. \\
\(\Omega_N\) & A narrow finite union of working bins used in the sharp source mean estimate. \\
\bottomrule
\end{tabular}
\end{table}

\subsection{Fixed parameters and uniformity}

Fix

\[
p=\frac54,\qquad q=\frac{13}{20},\qquad
A=1,\quad B=\frac{11}{10},\quad D=B-A=\frac1{10}.
\tag{1.3}\label{eq:1.3}
\]

Also fix

\[
\eta=\frac1{200},\quad \Pi=10^{-4},\quad
\rho_P=0.85,\quad \rho_E=0.01,\quad
\rho_I=0.86,\quad \rho_B=0.90.
\tag{1.4}\label{eq:1.4}
\]

Here \(\rho_I=\rho_P+\rho_E\). The geometric margins are supplied by

\[
d:=\frac{1+\log(B/A)}p<0.877,
\qquad d+\Pi<1;
\tag{1.5}\label{eq:1.5}
\]

the probabilistic margins are supplied by

\[
\frac1p=0.8<\rho_P,
\qquad 21\Pi=0.0021<\rho_E,
\qquad \rho_I<\rho_B<1.
\tag{1.6}\label{eq:1.6}
\]

For the endpoint reserve and the width bound on inherited jobs, we may take the specific constants

\[
c=\frac18,\qquad C=16.
\]

They satisfy

\[
C\ge \max\{A^{-1/p},(3/c)^{1/p}\}.
\tag{1.7}\label{eq:1.7}
\]

Indeed, \(16^{5/4}=32>24=3/c\). Once chosen, these constants do not vary with \(m\), the index, or the finite time horizon. Below, \(K\) denotes a positive constant depending only on these fixed parameters; its value may be increased from one occurrence to the next.

For nonnegative quantities \(X\) and \(Y\), we write \(X\asymp Y\) if there exist constants \(c_1,c_2>0\), depending only on the fixed parameters, such that $c_1Y\le X\le c_2Y$ throughout the stated range.
All comparison constants in \(\asymp\), all implicit constants in \(O(\cdot)\), and all convergence assertions in \(o(1)\) are \textbf{uniform} over finite construction histories satisfying the stated geometric conditions.
Only at the end do we increase \(m_0\) so that all uniform error terms are simultaneously smaller than their respective fixed margins.

To ''process index \(N\)'' means to count \(R_N\), at its reserved position, as a formally placed rectangle. A row may know all its assigned indices in advance, but its geometric decomposition is created only when its first index is processed. This timing convention is essential to the area identity in Section \ref{sec:7}.

\section{Admissible index intervals and full-height row geometry}\label{sec:2}

\textbf{Purpose of Step 2.} We first solve a local deterministic problem: given a suitable gap, any set of eligible indices of the prescribed size must fit. Randomness will decide which indices are assigned in Section \ref{sec:4}, but it will not be necessary to justify the resulting geometry.

For \(a>0\), define

\[
f_a(x)=\frac1{x+1}+a x^{-p},\qquad x>0.
\]

For a given row height \(H>0\), define the integer interval

\[
I(H)=\{i\in\mathbb N_{\ge1}:f_A(i)\le H\le f_B(i)\},
\qquad L(H)=|I(H)|.
\tag{2.1}\label{eq:2.1}
\]

An index \(i\) belongs to \(I(H)\) exactly when placing the rectangle of width \(1/i\) and height \(1/(i+1)\) with its top aligned to the top of a row of height \(H\) leaves a gap of height

\[
g_i=H-\frac1{i+1}
\]

below it, satisfying

\[
A i^{-p}\le g_i\le B i^{-p}.
\tag{2.2}\label{eq:2.2}
\]

\begin{lemma}[Uniform window estimates]
\label{thm:2.1}
As \(H\to0\), write \(s=H^{-1}\to\infty\). Then

\[
\min I(H)=s+O(s^{2-p}),
\qquad
L(H)=\bigl(D+O(s^{1-p})\bigr)s^{2-p}+O(1).
\tag{2.3}\label{eq:2.3}
\]

In particular, \(I(H)\) is nonempty for all sufficiently small \(H\). If

\[
a_{\max}=\frac1{\min I(H)},\qquad
a_{\min}=\frac1{\max I(H)},
\]

then

\[
(1-KH^{p-1})H\le a_{\min}\le a_{\max}\le H-cH^p,
\qquad a_{\max}-a_{\min}=O(H^p).
\tag{2.4}\label{eq:2.4}
\]

\end{lemma}

\begin{proof}
The function \(f_a\) is strictly decreasing.
Let \(x_a(H)\) be the unique positive solution of \(f_a(x)=H\). The equation first gives \(x_a\asymp s\); substituting this estimate back into the equation yields

\[
x_a=s+a s^{2-p}+O(s^{3-2p}+1),
\]

where the implicit constant and the threshold for \(s\) can be chosen independently of \(a\in[A,B]\).
Alternatively, differentiation with respect to the parameter gives

\[
\frac{\partial x_a}{\partial a}
=\frac{x_a^{-p}}{(x_a+1)^{-2}+pa x_a^{-p-1}}
=s^{2-p}\bigl(1+O(s^{1-p})\bigr).
\]

Consequently,

\[
I(H)=\mathbb N\cap[x_A(H),x_B(H)].
\]

Taking the ceiling at the left endpoint and the floor at the right endpoint changes the length by only \(O(1)\), proving \eqref{eq:2.3}. This also covers the case in which an endpoint is an integer.

Every admissible index satisfies \(i=s+O(s^{2-p})\), so \(1/i=H(1+O(H^{p-1}))\). Furthermore,

\[
H-\frac1i
=\left(H-\frac1{i+1}\right)-\frac1{i(i+1)}
\ge A i^{-p}-i^{-2}.
\]

Since \(p<2\), the right-hand side is \((A+o(1))H^p\), and is eventually at least \(H^p/2\), which exceeds the chosen reserve \(cH^p=H^p/8\). This proves the right bound in \eqref{eq:2.4}. The lower bound follows from
\[
\frac{1}{\max I(H)}=\frac{1}{\min I(H) + L(H)-1}\geq\frac1{s+Ks^{2-p}}=\frac{H}{1+KH^{p-1}}\geq (1-KH^{p-1})H,
\]
where the last inequality uses \((1+u)^{-1}\ge 1-u\) for \(u\to0\).

Finally, the window length \(O(s^{2-p})\), together with the derivative bound for the reciprocal function, gives a width variation of \(O(s^{-p})=O(H^p)\). This proves the lemma.

\[
a_{\max}-a_{\min}=\frac{\max I(H)-\min I(H)}{\max I(H)\cdot\min I(H)}=O(s^{-p})=O(H^p).
\]

\end{proof}

An \textbf{inherited row job} is an unused rectangle \(W\times H\) satisfying

\[
H\le W\le C H^{1/p}.
\tag{2.5}\label{eq:2.5}
\]

Its demand is defined by

\[
k(W,H)=\left\lfloor\frac{W-cH^p}{a_{\max}}\right\rfloor.
\tag{2.6}\label{eq:2.6}
\]

A job must receive \(k(W,H)\) distinct indices from \(I(H)\).

\begin{lemma}[Every admissible index set fits]
\label{thm:2.2}
For all sufficiently small \(H\),

\[
1\le\left\lfloor\frac WH\right\rfloor\le k(W,H)
\le K s^{1-1/p}=K s^{1/5}.
\tag{2.7}\label{eq:2.7}
\]

Choose any \(k=k(W,H)\) distinct indices \(i\in I(H)\). The corresponding rectangles fit into \(W\times H\) when placed side by side in increasing index order, with their tops aligned. The endpoint width

\[
e=W-\sum_{i\text{ selected}}\frac1i
\]

satisfies

\[
cH^p\le e\le H+O(W H^{p-1})=(1+o(1))H.
\tag{2.8}\label{eq:2.8}
\]

\end{lemma}

\begin{proof}
Since \(W\ge H\) and \(a_{\max}\le H-cH^p\),

\[
W(H-a_{\max})\ge HcH^p,
\quad\text{and hence}\quad
\frac{W-cH^p}{a_{\max}}\ge\frac WH.
\]

Taking floors gives the lower bound on the demand. Also, \(a_{\max}\ge H/2\) eventually, so

\[
k\le \frac W{a_{\max}}\le 2C H^{1/p-1}.
\]

This proves \eqref{eq:2.7}. Because \(1/5<3/4\), we also have \(k\le L(H)\); thus an individual job never requires more distinct indices than its window contains.

Each selected rectangle has width at most \(a_{\max}\), so the total width is at most \(k a_{\max}\le W-cH^p\). It follows that \(e\ge cH^p>0\). The definition using the floor also gives

\[
W-cH^p<(k+1)a_{\max}.
\]

Therefore,

\[
e<cH^p+a_{\max}+k(a_{\max}-a_{\min})
\le H+O(W H^{p-1}).
\]

Finally,

\[
\frac{W H^{p-1}}H\le C H^{p+1/p-2}
=C H^{1/20}\longrightarrow0.
\]

This proves all the assertions, without requiring the selected indices to be consecutive.

\end{proof}

The exact decomposition of the row consists of the selected rectangles, one gap of dimensions \((1/i)\times g_i\) below each rectangle, and one endpoint gap of dimensions \(H\times e\). Gaps of the first type are called \textbf{principal gaps}. By \eqref{eq:2.2},

\[
\frac1i\le A^{-1/p}g_i^{1/p}\le Cg_i^{1/p}.
\]

Moreover, \(g_i<1/i\) eventually, so every principal gap is itself an inherited job satisfying \eqref{eq:2.5}. Whether the parent row is a main row or an endpoint row, each index produces exactly one such principal gap.

\paragraph{The actual regions in the row.}
In local coordinates let the parent be \([0,W]\times[0,H]\), and list its assigned indices as \(i_0,\ldots,i_{k-1}\). Put
\[
 a_j=1/i_j,\quad b_j=1/(i_j+1),\quad
 x_j=\sum_{\ell<j}a_\ell,\quad g_j=H-b_j,\quad e=W-x_k.
\]
The closed tile, the closed principal gap below it, and the endpoint region are
\[
\begin{aligned}
 P_j&=[x_j,x_{j+1}]\times[g_j,H],\\
 G_j&=[x_j,x_{j+1}]\times[0,g_j],\\
 E&=[x_k,W]\times[0,H].
\end{aligned}
\]
Consequently
\[
 [0,W]\times[0,H]=\left(\bigcup_{j<k}(P_j\cup G_j)\right)\cup E.
\]
Different columns have disjoint horizontal interiors. Within a column, the tile and gap meet on their shared boundary only. The endpoint has disjoint interior from every column. The endpoint split preserves this exact closed-set coverage and interior disjointness. This verifies the geometric decomposition itself, in addition to its numerical area identity.

\section{Endpoint splitting and the total boundary budget}\label{sec:3}

The principal gaps already satisfy the same shape condition as their parents. The endpoint gap needs a separate split. One generation need not contract its total short side; the useful contraction appears only after two generations.

Every row splits its endpoint rectangle \(H\times e\) into

\[
H\times(\theta e),\qquad H\times((1-\theta)e),
\qquad \theta\sim\operatorname{Unif}[1/3,2/5].
\tag{3.1}\label{eq:3.1}
\]

The row's parameter \(\theta\) is sampled independently only after \(H\) and \(e\) have been determined. The two children share the same \(\theta\); independence between siblings is not assumed. Each child becomes an \textbf{endpoint job}, whose later endpoint is split according to the same rule.

\begin{lemma}[Closure of the geometric conditions]
\label{thm:3.1}
For all sufficiently small heights, both children produced by a split satisfy \eqref{eq:2.5}. Their shorter sides are at most \((3/4)H\).

\end{lemma}

\begin{proof}
By \eqref{eq:2.8}, we may uniformly require \(e\le(1+\eta)H\). The two shorter sides, \(h_1=\theta e\) and \(h_2=(1-\theta)e\), satisfy

\[
\frac c3H^p\le h_j\le\frac23(1+\eta)H<\frac34H<H.
\]

Their longer side is therefore indeed \(H\), and

\[
H\le(3/c)^{1/p}h_j^{1/p}\le C h_j^{1/p}.
\]

Thus Lemma~\ref{thm:2.2} may be applied recursively.

\end{proof}

We next establish a budget for the endpoint trees by summing shorter sides. This controls both the heights used at the cuts and the area of waiting endpoint rectangles.

\begin{lemma}[Contraction over two generations]
\label{thm:3.2}
Consider an endpoint job with shorter side \(H\). The sum of the shorter sides in the next generation is at most \((1+\eta)H\), and the corresponding sum two generations later is at most \(rH\), where

\[
r=\frac{6(1+\eta)}{7-\eta}=\frac{1206}{1399}<0.87.
\tag{3.2}\label{eq:3.2}
\]

\end{lemma}

\begin{proof}
The sum of the shorter sides in the next generation is exactly \(e\le(1+\eta)H\). Let \(e_1,e_2\) be the endpoint widths left by the two children. The sum of the grandchildren's shorter sides is \(E=e_1+e_2\). First,

\[
E\le(1+\eta)(h_1+h_2)=(1+\eta)e.
\tag{3.3}\label{eq:3.3}
\]

Since \(h_1\le(2/5)(1+\eta)H<H/2\), the first child has width-to-height ratio greater than 2, so Lemma~\ref{thm:2.2} guarantees that it receives at least two rectangles. The second receives at least one. Below the chosen height threshold, every rectangle has width at least \((1-\eta)\) times the height of its row. Hence the total horizontal width removed from the two children is at least

\[
(1-\eta)(2h_1+h_2)
=(1-\eta)(1+\theta)e
\ge\frac43(1-\eta)e.
\]

The sum of the children's original widths is \(2H\), and therefore

\[
E\le2H-\frac43(1-\eta)e.
\tag{3.4}\label{eq:3.4}
\]

Set \(x=e/H\). Combining \eqref{eq:3.3} and \eqref{eq:3.4} gives

\[
\frac EH\le
\min\left\{(1+\eta)x,\ 2-\frac43(1-\eta)x\right\}.
\]

One affine function is increasing and the other decreasing, so the maximum of their minimum occurs at their intersection. At that point,

\[
\left(1+\eta+\frac43(1-\eta)\right)x=2.
\]

Substitution gives \eqref{eq:3.2}. The proof also covers an original job that receives only one rectangle, an integral width-to-height ratio, and an arbitrarily small positive endpoint width \(e\).

\end{proof}

\begin{corollary}[Endpoint forest budget]
\label{thm:3.3}
If the sum of the root shorter sides of an endpoint forest is \(S_0\), then the sum of the shorter sides of all nodes that have been created is at most

\[
\frac{2+\eta}{1-r}S_0\le16S_0.
\tag{3.5}\label{eq:3.5}
\]

The last inequality may be made strict when \(S_0>0\). For the empty forest, \(S_0=0\), both sides are zero; no strict inequality is asserted.

Call a row cut directly from the primary rectangle, or a row filling a principal gap, an \textbf{ordinary row}. For an ordinary row of height \(H\), together with all its endpoint descendants, the sum of the parent-row heights used in all endpoint splits is less than \(18H\), and the sum of the semiperimeters of all endpoint rectangles created is less than \(55H\). In particular, both sums are less than \(125H\).

\end{corollary}

\begin{proof}
Sum by generations. The even generations are bounded by \(r^jS_0\), and the odd generations by \((1+\eta)r^jS_0\), which gives \eqref{eq:3.5}. The first split of an ordinary row creates two roots whose shorter sides sum to \(e\le(1+\eta)H\). Including the cut made by the ordinary row itself, the sum of all cut-parent heights is less than

\[
H+16e\le\bigl(1+16(1+\eta)\bigr)H<18H.
\]

A cut with parent height \(h\) creates two endpoint rectangles whose semiperimeters sum to \(2h+e'\le(3+\eta)h\). Summing gives a bound of less than \((3+\eta)18H<55H\).

For an incomplete finite tree, assign zero contribution to grandchildren that have not yet been created. Estimate \eqref{eq:3.3} plainly remains valid. Estimate \eqref{eq:3.4} also remains valid: for an unprocessed first or second child, respectively, its zero contribution may be bounded by the nonnegative quantities \(H-2(1-\eta)h_1\) and \(H-(1-\eta)h_2\). These quantities are nonnegative because \(h_1<H/2\) and \(h_2<H\), respectively. Thus the two-generation inequalities apply directly to truncated trees. They do not require future jobs to have feasible assignments, nor do they presuppose the existence of an infinite packing.

\end{proof}

\section{Local integer quotas}\label{sec:4}

\textbf{Allocation mechanism for Step 4.} A real-valued demand fraction is not an actual tile assignment. The following cumulative rounding makes the demands integer, preserves each job's total demand exactly, and leaves a positive density of main positions when each bin certificate is good.

Starting from \(s_0=m\), define consecutive bins of integer indices by

\[
\mathcal B_b=\{s_b,\ldots,s_b+\ell_b-1\},\qquad
\ell_b=\lfloor s_b^q\rfloor,\quad s_{b+1}=s_b+\ell_b.
\tag{4.1}\label{eq:4.1}
\]

In this section, a ``bin'' is an interval of integer indices, not a geometric rectangle.

For job \(r:W_r\times H_r\), let \(I_r\) denote its window, \(L_r\) its length, and \(k_r\) its demand. Retain only the bins contained entirely in \(I_r\). At most two boundary-bin portions are discarded, so the retained length \(L'_r\) satisfies

\[
L'_r=L_r-O(s_r^q),\qquad
\frac{L'_r}{L_r}=1-O(s_r^{-1/10}),\quad s_r=H_r^{-1}.
\tag{4.2}\label{eq:4.2}
\]

In particular, every sufficiently large window contains complete bins, and \(L'_r>0\). List the retained bins as \(b_1,\ldots,b_J\), and put \(C_j=\sum_{a\le j}\ell_{b_a}\), with \(C_0=0\). Associate with this job an independent random variable \(V_r\sim\operatorname{Unif}[0,1)\), and define

\[
t_{r,b_j}
=\left\lfloor\frac{k_rC_j}{L'_r}+V_r\right\rfloor
-\left\lfloor\frac{k_rC_{j-1}}{L'_r}+V_r\right\rfloor.
\tag{4.3}\label{eq:4.3}
\]

The quota is zero for bins that were not retained.

\begin{lemma}[Randomized rounding]
\label{thm:4.1}
All quotas are nonnegative integers, and

\[
\sum_b t_{r,b}=k_r,\qquad
\mathbb E t_{r,b}=\frac{k_r\ell_b}{L'_r},\qquad
t_{r,b}\le\frac{k_r\ell_b}{L'_r}+1.
\tag{4.4}\label{eq:4.4}
\]

Changing \(V_r\) changes the quota in any one bin by at most 1.

\end{lemma}

\begin{proof}
Nonnegativity follows because the cumulative quantities are increasing. The total telescopes to \(\lfloor k_r+V_r\rfloor-\lfloor V_r\rfloor=k_r\). For every real number \(x\), we have \(\mathbb E\lfloor x+V_r\rfloor=x\); taking differences proves the expectation formula. A floor increment over a real interval of length \(u\) is at most \(u+1\). Finally, if \(V'_r\ge V_r\), each cumulative floor increases by either 0 or 1, so the difference between two successive increases lies in \(\{-1,0,1\}\).

\end{proof}

The case \(k_r=1\) is allowed. Rounding errors cannot be absorbed by assuming that the quota in each bin is large; endpoint jobs may continue to have demand equal to one.

Define the fractional load of a fixed pool of jobs by

\[
F(i)=\sum_{r:i\in I_r}\frac{k_r}{L_r}.
\tag{4.5}\label{eq:4.5}
\]

This sum uses the \textbf{reservation ledger}: once a job is created, its record \((W,H,I,k)\) is retained. Even after the row has started, its summand remains included whenever \(i\in I_r\). By contrast, the area estimates in Sections \ref{sec:6} and \ref{sec:7} use the \textbf{geometric ledger}, which counts only gaps whose rows have not yet started. The two ledgers must be kept distinct. The principal-gap jobs and endpoint jobs in the reservation ledger define \(F_P(i)\) and \(F_E(i)\), respectively, so that \(F=F_P+F_E\).

\begin{lemma}[Bin capacity]
\label{thm:4.2}
Suppose that the fixed pool of jobs near a bin satisfies \(F(i)\le\rho_I\). Then, for that bin,

\[
\mathbb P\left\{\sum_r t_{r,b}\ge\rho_B\ell_b\right\}
\le K\exp(-\kappa s_b^{11/20}).
\tag{4.6}\label{eq:4.6}
\]

Here and below, the positive constant \(\kappa\) may be decreased from one occurrence to the next.

\end{lemma}

\begin{proof}
By \eqref{eq:4.2} and \eqref{eq:4.4},

\[
\mu_b:=\mathbb E\sum_r t_{r,b}
\le(1+O(s_b^{-1/10}))\rho_I\ell_b.
\]

Quotas belonging to different jobs are independent. Every quota is at most

\[
K s_b^{1/5+13/20-3/4}+1\le K s_b^{1/10}=:M_b.
\]

A nonnegative random variable bounded by \(M_b\) satisfies \(\operatorname{Var}(X)\le M_b\mathbb EX\). Hence the total variance is at most \(M_b\mu_b=O(s_b^{3/4})\). The threshold exceeds the mean by a fixed positive multiple of \(\ell_b\). Bernstein's inequality gives an exponent of order

\[
\frac{\ell_b^2}{K s_b^{3/4}+K s_b^{1/10}\ell_b}
\asymp s_b^{13/10-3/4}=s_b^{11/20}.
\]

The probability inequality used here is proved in Section \ref{sec:8}.

\end{proof}

When the total quota is at most \(\rho_B\ell_b\), put the corresponding job labels into the bin and fill the remaining positions with ``main position'' labels. Temporarily regard all label copies as distinct objects, and permute them uniformly at random; use independent permutations in different bins. Every job then receives exactly \(k_r\) distinct indices, there are no index conflicts, and every bin has at least \(0.1\ell_b\) main positions. Discarding the temporary copy identifiers gives the required assignment.

The random variable \(V_r\) is shared across the several bins of the same job; \textbf{quotas in different bins are not independent}. We use only independence of \(V_r\) across jobs and, conditional on all quotas, independence of the permutations across bins.

\section{The geometry of main rows}\label{sec:5}

\textbf{Beginning Step 3.} An index left unassigned to inherited jobs is labeled as ''main position''. We use it to start a strip in the primary rectangle. The membership of that strip must be settled before its height is randomized; this is the geometric source of predictability used later.

Maintain one unused rectangle \(P\), called the \textbf{primary rectangle}. Initially, \(P=Q_m\). It is reduced only by removing full strips spanning its shorter side.

Let \(S\) be the next unprocessed main position, and let \(w\) be the current shorter side of the primary rectangle. Temporarily assume that

\[
a_0 S^{-1/2}\le w\le S^{-1/2},
\tag{5.1}\label{eq:5.1}
\]

where \(a_0>0\) is fixed; Section \ref{sec:7} proves that this condition can be maintained. Fix \(c_0=1/4\), and set \(r_S=c_0S^{-p}\). Read the main positions in order, taking the longest initial segment satisfying

\[
\sum_{j=1}^k\frac1{i_j}\le w-r_S.
\tag{5.2}\label{eq:5.2}
\]

Here \(i_1=S\) is the starting index, and the row height is sampled only after all members of the row have been selected.

\begin{lemma}[Main-row size and height]
\label{thm:5.1}
Suppose that main positions occupy at least one tenth of each bin from Lemma~\ref{thm:4.2} and that \eqref{eq:5.1} holds. For all sufficiently large \(S\), the preceding procedure satisfies

\[
k\asymp\sqrt S,\qquad
i_j=S+O(S^q),\qquad
r_S\le e:=w-\sum_{j=1}^k\frac1{i_j}
<r_S+\frac1{i_{k+1}}.
\tag{5.3}\label{eq:5.3}
\]

All selected indices lie in the remainder of the current bin and the immediately following bin. Set \(z=S\), and let

\[
\varepsilon_S=K_0S^{-1/10},\qquad
U\sim\operatorname{Unif}[A+\varepsilon_S,B-\varepsilon_S],
\qquad
H=\frac1{z+1}+U z^{-p},
\tag{5.4}\label{eq:5.4}
\]

where \(K_0\) is fixed and sufficiently large. Then every selected rectangle satisfies \eqref{eq:2.2}, and the endpoint satisfies

\[
cH^p\le e\le(1+o(1))H.
\tag{5.5}\label{eq:5.5}
\]

\end{lemma}

\begin{proof}
The next complete bin has at least \(a_1S^q\) main positions, each of width at least \((1-o(1))/S\). Their total width is at least \(a_2S^{q-1}\). Since \(q>1/2\), this exceeds \(w\le S^{-1/2}\). Thus the greedy selection must stop before exhausting that bin. This proves the span bound \(O(S^q)\) independently, without using the batch size to assume the span bound that is needed to establish that size.

Within this range, every rectangle has width \((1+o(1))/S\). By maximality of the initial segment, the endpoint width is less than \(r_S+1/i_{k+1}=O(1/S)\). Consequently, the occupied width is \(w-O(1/S)\asymp S^{-1/2}\), giving \(k\asymp\sqrt S\). In particular, \(k\ge1\), and the next main position exists.

For a selected index \(i\),

\[
\left|\frac1{z+1}-\frac1{i+1}\right|=O(S^{q-2}),
\qquad (i/z)^p=1+O(S^{q-1}).
\]

Hence

\[
i^p\left(H-\frac1{i+1}\right)
=U+O(S^{p+q-2})=U+O(S^{-1/10}).
\tag{5.6}\label{eq:5.6}
\]

Choosing \(K_0\) larger than the constant in this uniform error bound gives the exact interval condition \eqref{eq:2.2}. For all sufficiently large \(S\), we have \(2\varepsilon_S<D\), so the sampling interval is nonempty. We may also require \(H\le3/(2S)\), in which case

\[
cH^p\le\frac18(3/2)^{5/4}S^{-p}<\frac14S^{-p}=r_S.
\]

For the upper bound, the stronger estimate

\[
e<\frac1S+\frac14S^{-p}
\le\frac1{S+1}+S^{-p}\le H
\]

eventually holds, since \(S^{-p}\gg S^{-2}\). This proves \eqref{eq:5.5}.

\end{proof}

Remove the entire \(w\times H\) strip from the primary rectangle, reducing its longer side by \(H\), and place the main row in the strip. Within the strip, the only new gaps are the principal gap below each rectangle and one endpoint, which is split according to \eqref{eq:3.1}. No additional unaccounted full-width thin strip is left below the row.

\begin{lemma}[Aspect ratio]
\label{thm:5.2}
Suppose that the primary rectangle initially has aspect ratio at most 2 and that \(H\le w/2\). After removing the strip just described, the remaining primary rectangle still has aspect ratio at most 2.

\end{lemma}

\begin{proof}
Write the original dimensions as \(w\le h\le2w\). If \(h-H\ge w\), the aspect ratio does not increase. If \(h-H<w\), then \(h-H\ge w-H\ge w/2\), so after interchanging the longer and shorter sides the ratio is still at most 2. Equality cases are included in these inequalities.

\end{proof}

\section{An area bound for waiting principal gaps}\label{sec:6}

This estimate answers an area question only: how much space can be tied up in small principal gaps? It uses no independence, no density assumption for their creation, and no successful future assignment.

\begin{lemma}[Deterministic area estimate]
\label{thm:6.1}
Consider any family of principal gaps satisfying \eqref{eq:2.2}, with at most one gap associated with each index. Let \(D_P(t)\) be the total area of those gaps whose heights satisfy \(g_i\le t\). Then, as \(t\downarrow0\),

\[
D_P(t)\le d\,t+O(t^{1+1/p}),
\qquad d=\frac{1+\log(B/A)}p.
\tag{6.1}\label{eq:6.1}
\]

\end{lemma}

\begin{proof}
Every index included in the sum satisfies \(i\ge(A/t)^{1/p}\), and the area of its gap is at most

\[
\frac1i\min\{B i^{-p},t\}.
\]

Set \(a=(A/t)^{1/p}\) and \(b=(B/t)^{1/p}\). Enlarging the sum to include all eligible integers, rather than only gaps already generated, gives

\[
D_P(t)
\le t\sum_{a\le i\le b}\frac1i
+B\sum_{i>b}i^{-p-1}.
\tag{6.2}\label{eq:6.2}
\]

An integer term at the dividing point may be assigned to either sum, since the two expressions agree there. Any extra term introduced by rounding has size \(O(t/a)\). Integral comparison for monotone functions yields

\[
\sum_{a\le i\le b}\frac1i
\le\log(b/a)+O(a^{-1}),
\quad
\sum_{i>b}i^{-p-1}\le\frac{b^{-p}}p+O(b^{-p-1}).
\]

Substituting \(\log(b/a)=p^{-1}\log(B/A)\) and \(Bb^{-p}=t\) proves \eqref{eq:6.1}. The estimate permits arbitrary dependence between gap heights and indices, and remains valid after deleting gaps that have already been filled or have not yet been generated.

\end{proof}

Immediately before index \(N\) is processed, an unstarted inherited row job with a valid assignment has its first assigned index \(j\ge N\). Consequently,

\[
H\le f_B(j)\le f_B(N)=:t_N
=\frac{1+o(1)}N.
\tag{6.3}\label{eq:6.3}
\]

The total area of waiting principal gaps is therefore at most \((d+o(1))/N\). A job whose reveal time has not yet been reached satisfies the same bound: otherwise its entire admissible window would precede the current time, so its assignment would already have been required and would constitute an earlier failure. This observation is used only before the first failure.

\section{Area accounting and the geometric bootstrap}\label{sec:7}

The next identity joins the local constructions into a global invariant. The crucial distinction is between an entire row already reserved in space and the smaller set of tile indices already processed.

The construction starts from an empty square, with no pre-existing gaps. Whenever a row is started, its parent rectangle is replaced by the placed or reserved rectangle positions in that row, all its principal gaps, and its two endpoint children. The primary rectangle loses one complete strip only when a main row is started. These operations maintain an exact decomposition into regions with pairwise disjoint interiors.

Immediately before index \(N\) is processed, the rectangles with indices \(m,\ldots,N-1\) have been placed. The remaining area is exactly

\[
\frac1m-\sum_{i=m}^{N-1}\frac1{i(i+1)}=\frac1N.
\tag{7.1}\label{eq:7.1}
\]

This area has four components: the primary rectangle, unstarted principal-gap jobs, unstarted endpoint jobs, and the rectangle positions reserved for indices \(i\ge N\) in rows that have already started. A started parent job is no longer counted among the waiting jobs; each of its children is counted according to whether that child has started. The principal gap below a reserved future rectangle is generated when its row starts and belongs to the second component, while the reserved rectangle itself belongs to the fourth. Thus neither region is omitted or counted twice.

There is also a cumulative endpoint account: processed endpoint nodes remain in its historical forest. A split at height \(H\) and endpoint width \(e\) adds \(H\) to the cut-height total and \(2H+e\) to the total semiperimeter of created endpoints. It does not subtract the old parent's semiperimeter. Waiting endpoints form a subfamily of all created endpoints, so their total area can be bounded by this larger nonnegative account.

For reference, the four-region equality underlying \eqref{eq:7.1} is
\[
 A_{\mathrm{primary}}(N)+A_{\mathrm{waiting\ principal}}(N)
 +A_{\mathrm{waiting\ endpoint}}(N)+A_{\mathrm{reserved}}(N)=1/N.
\]
It follows from the exact region decomposition; it is not an extra condition assumed about a successful future state.

\begin{lemma}[Reserved area]
\label{thm:7.1}
In any finite history before failure, all indices of reserved rectangles that have not yet been used lie in

\[
[N,N+K N^{3/4}].
\]

Their total area is therefore at most \(K N^{-5/4}\).

\end{lemma}

\begin{proof}
An inherited row starts only at its first assigned index, and all its assigned indices belong to a single window of length \(O(s^{3/4})\). If a row extends across the current index \(N\), then its scale satisfies \(s\asymp N\), and its last and first indices differ by \(O(N^{3/4})\). Main rows have the smaller span \(O(N^{13/20})\). All assigned indices are distinct. Hence, regardless of how many rows are simultaneously open, the sum of the areas of all rectangles indexed by this integer interval is a uniform upper bound:

\[
\sum_{i=N}^{N+\lceil K N^{3/4}\rceil}\frac1{i(i+1)}
\le K' N^{-5/4}.
\]

This is the role of delaying the geometric start of a row. If the complete row were materialized at a time earlier by a fixed proportion, the range of reserved indices could have length \(\Theta(N)\), and the preceding lower-order bound would no longer hold.

\end{proof}

\begin{lemma}[Simultaneous maintenance of the budgets and the primary rectangle]
\label{thm:7.2}
Suppose that all fractional-load and bin-capacity certificates required so far are valid. For a sufficiently large initial index \(m\), the following properties hold throughout every finite horizon:

\begin{enumerate}
\item The sum of the semiperimeters of all endpoint rectangles generated is less than \(\Pi/2\), and the sum of the parent heights used in all endpoint cuts is also less than \(\Pi/2\).
\item The aspect ratio of the primary rectangle is at most \(2\), and its area \(A_N\), immediately before index \(N\), satisfies
\end{enumerate}

\[
\frac{1-d-\Pi-o(1)}N\le A_N\le\frac1N.
\tag{7.2}\label{eq:7.2}
\]

\end{lemma}

\begin{proof}
We use a first-failure argument. As provisional inductive conditions, require that each of the two endpoint budgets be at most \(\Pi\), that the shorter side of the primary rectangle be at least \(a_0N^{-1/2}\), and that its aspect ratio be at most \(2\), where

\[
0<a_0<\sqrt{(1-d-\Pi)/4}.
\tag{7.3}\label{eq:7.3}
\]

On a history in which these conditions have held, a main row uses \(\asymp\sqrt S\) indices. Main rows whose initial indices lie in \([T,2T]\) use indices in \([T,2T+KT^q]\). Each such row uses at least \(a_3\sqrt T\) distinct indices, so there are at most \(K\sqrt T\) such rows. Their individual heights are \(O(1/T)\). Summing over the intervals \([2^jm,2^{j+1}m]\) gives

\[
\sum_{\text{main rows}}H\le K\sum_{j\ge0}(2^jm)^{-1/2}
\le K' m^{-1/2}.
\tag{7.4}\label{eq:7.4}
\]

Each principal-gap row corresponds to a unique source index \(i\ge m\). This index is not counted twice even if its rectangle has only been reserved. Thus

\[
\sum_{\text{principal-gap rows}}H
\le B\sum_{i\ge m}i^{-p}
\le K m^{1-p}=K m^{-1/4}.
\tag{7.5}\label{eq:7.5}
\]

By Corollary 3.3, both endpoint budgets are at most

\[
125K(m^{-1/2}+m^{-1/4}).
\tag{7.6}\label{eq:7.6}
\]

Increasing \(m\) makes this quantity smaller than \(\Pi/2\). The bound is uniform over all finite prefixes. It also includes a new row that could potentially cause the first budget violation: its number of members and its height are determined from the valid state before the cut, after which the row is included in the count in \eqref{eq:7.4} or \eqref{eq:7.5}. Hence neither budget can fail first.

Before index \(N\), every waiting endpoint has shorter side at most \(t_N\). If its longer side is \(W\), its area is at most \(t_NW\). The endpoint semiperimeter budget therefore bounds the total area of waiting endpoints by \(\Pi t_N\). Combining Lemmas~\ref{thm:6.1} and \ref{thm:7.1} with the exact area identity \eqref{eq:7.1}, we obtain

\[
A_N\ge\frac1N-(d+o(1))\frac1N-\Pi t_N-KN^{-5/4},
\]

which is \eqref{eq:7.2}. The upper bound follows directly from the decomposition into regions of nonnegative area.

When the aspect ratio is at most \(2\), the shorter side of the primary rectangle satisfies

\[
\sqrt{A_N/2}\le w_N\le\sqrt{A_N}.
\]

Equation \eqref{eq:7.2} consequently provides a lower bound strictly stronger than the inductive lower bound chosen in \eqref{eq:7.3}. A main row has height \(O(1/N)=o(w_N)\), so its strip fits, and Lemma~\ref{thm:5.2} preserves the aspect ratio. Lemmas~\ref{thm:2.2} and \ref{thm:3.1} preserve the geometric conditions for all inherited rows and new children. Every main-row height is at most \(3/(2m)\); the heights of inherited row jobs are at most \(f_B(m)\); and endpoint-child heights decrease strictly. Increasing \(m\) places all these heights below a common small-height threshold, so all geometric lemmas apply uniformly. Thus no geometric condition can be the first to fail.

Initially there is no waiting or reserved area, \(A_m=1/m\), the aspect ratio is \(1\), and both budgets are zero. The induction therefore has a valid initial state.

\end{proof}

\section{Principal-gap load for a fixed input pool}\label{sec:9}

Fix a target integer \(N\), and write

\[
S=N^{1/p}=N^{4/5}.
\]

\textbf{What this section proves.} Freeze the geometries and demands of a finite pool of inherited jobs supplying a narrow set of working bins. Assume its fractional input load is at most \(\rho_I\) on a slightly larger comparison set containing their complete windows. Define \(Z_N\) to be the inherited output load on the working positions plus a deterministic weight at every remaining main position. We first prove the \emph{source-only} estimate
\[
\mathbb P\{\text{some working bin has occupancy }\ge0.90\ell_b
                 \ \text{or}\ Z_N\ge0.83\}
 \le3e^{-\kappa_sN^{7/50}},\qquad \kappa_s>0.
\tag{8.1}\label{eq:9.1}
\]
The constants and the large-scale threshold are uniform over these pools. If quotas exceed the available capacity, there is no physical allocation; an arbitrary measurable extension of \(Z_N\) may be used there, since the bin event already includes that outcome.

The other random component is the height of each actual main row. Section~\ref{sec:9.7} gives a separate estimate for its accumulated load minus used comparison weight. That estimate is taken in the actual query filtration. It does not condition on all future bin allocations. Sections~\ref{sec:11}--\ref{sec:12} justify the freezing and combine the two bounds.

\subsection{A single output gap and its deterministic support}\label{sec:9.1}

A source index \(i\) generates a principal gap of width \(1/i\) and height \(g\). Its load coefficient at \(N\) is

\[
\lambda_N(i,g)
=\frac{k(1/i,g)}{L(g)}\,\mathbf1_{\{N\in I(g)\}}.
\tag{8.2}\label{eq:9.2}
\]

If \(N\in I(g)\), then

\[
f_A(N)\le g\le f_B(N),\qquad g=(1+o(1))/N.
\]

Lemma~\ref{thm:2.1} and the demand upper bound imply

\[
\lambda_N(i,g)
\le\bigl(1+O(N^{-1/4})\bigr)\frac{N^{p-1}}{Di}.
\tag{8.3}\label{eq:9.3}
\]

Only the inequality \(k(1/i,g)\le\min I(g)/i\) is used here. Taking the integer part in the demand and reserving a positive endpoint width therefore cannot increase this upper bound.

The exact gap profile \(Ai^{-p}\le g\le Bi^{-p}\) shows that every source index with nonzero contribution belongs to the deterministic interval

\[
K_N=\left[
\left(\frac A{f_B(N)}\right)^{1/p},
\left(\frac B{f_A(N)}\right)^{1/p}
\right].
\tag{8.4}\label{eq:9.4}
\]

Its endpoints are \(A^{1/p}S+O(S^{11/16})\) and \(B^{1/p}S+O(S^{11/16})\), respectively. In particular, \(i\asymp S\), and hence

\[
0\le\lambda_N(i,g)\le h_N:=K N^{-11/20}.
\tag{8.5}\label{eq:9.5}
\]

\subsection{A deterministic mean envelope for main positions}\label{sec:9.2}

The complete membership of a main row is determined before that row's own height mark is revealed. Let \(z\) be its reference index, and recall the height margin \(\varepsilon_z\) from Section~\ref{sec:5}. For a member \(i\), the conditional density of \(g=H-1/(i+1)\) is at most
\[
 \frac{z^p}{D-2\varepsilon_z}.
\]
Since \(i-z=O(z^q)\) and \(\varepsilon_z\to0\), this divided by \(i^p/D\) tends uniformly to one. The eligible height interval has length exactly \(DN^{-p}\). Multiplying its length by the density bound and \eqref{eq:9.3} shows the following uniform statement: for every fixed \(\epsilon>0\), all sufficiently large relevant scales satisfy
\[
 \mathbb E[\lambda_N(i,g)\mid\text{group-selection history}]
 \le (1+\epsilon)\frac{i^{p-1}}{DN}.
\tag{8.6}\label{eq:9.6}
\]
Positions in one main row share a height. We use linearity of conditional expectation for this common height, not independence of the positions.

Choose a fixed buffer large enough to contain the entire retained window of each inherited job capable of positive output at \(N\), and round its ends out to complete bins. Denote the resulting narrow union by \(\Omega_N\). It contains every actual source position \(i\ge m\) in \(K_N\), and is contained in
\[
 [S-K_1S^{3/4},\ B^{4/5}S+K_1S^{3/4}]\cap[m,\infty)
\]
for a fixed \(K_1\). The actual construction permits an inner buffer coefficient four before rounding and an envelope coefficient six after rounding. Complete retained windows of boundary input jobs are covered by a separate comparison envelope with coefficient eight. These buffers are fixed before increasing the scale.

Keep \emph{all} jobs that supply a working bin, including those with zero output at \(N\). Such a zero-output boundary job may have part of its retained window outside \(\Omega_N\). Only a job with positive output is required to have its whole retained window inside \(\Omega_N\). The larger comparison set controls all input loads and the number of jobs; it is not substituted for \(\Omega_N\) in the sharp mean sum.

Fix from now on
\[
 \epsilon=10^{-3},\qquad
 w_i=(1+\epsilon)\frac{i^{p-1}}{DN}
           \mathbf1_{\{i\in\Omega_N\}}.
\tag{8.7}\label{eq:9.7}
\]
All thresholds below are increased so that this one fixed margin simultaneously bounds the main-position mean and the inherited-row comparison proved next. There is no asserted rate for the threshold's dependence on \(\epsilon\). In particular, this fixed weight is not being written as \((1+o(1))i^{p-1}/(DN)\).

If the nominal interval begins below \(m\), only its actual positions \(i\ge m\) must be covered. All tile selections have indices at least \(m\); missing smaller integers produce no query load, no used weight, and no inherited source output. Dropping their nonnegative terms from an upper bound is valid without an assumption that the nominal left endpoint exceeds \(m\). After enlarging \(K\), both \(w_i\) and every source label cost are bounded by \(h_N\) from \eqref{eq:9.5}.

\subsection{The mean contribution of one inherited row job}\label{sec:9.3}

Fix a job \(r\), and write \(H_r=1/s_r\). A source position of this job can produce a nonzero output only if

\[
H_r-\frac1{i+1}\in[f_A(N),f_B(N)].
\]

With \(g=H_r-1/(i+1)\) as the variable, the inverse function and its derivative are

\[
i(g)=\frac1{H_r-g}-1,\qquad i'(g)=(i(g)+1)^2.
\]

Inside the window, \(i=s_r+O(S^{3/4})\). The number of eligible integer positions is therefore at most

\[
\bigl(D+O(S^{-1/4})\bigr)s_r^2N^{-p}+O(1).
\tag{8.8}\label{eq:9.8}
\]

The principal term \(s_r^2N^{-p}\asymp S^{7/16}\) tends to infinity, so the integer-endpoint error is relatively negligible. This count is unaffected when the eligibility interval crosses a bin boundary; restricting it to the complete bins retained by the job can only reduce the number.

On the full retained region \(I'_r\), define the formal position mass

\[
q_{r,i}=k_r/L'_r,
\]

and set it equal to zero outside that region. For a supported working bin, cumulative rounding has mean \(k_r\ell_b/L'_r\), and a uniform permutation distributes that quota equally among its positions. Their algebraic combination is this position mass. This calculation averages quota marks without conditioning on the event that all quotas fit. The algebraic quantity remains meaningful on outcomes with excessive quotas; we do not require an actual permutation to exist for every outcome.

Combining \eqref{eq:4.2}, \eqref{eq:9.3}, and \eqref{eq:9.8}, the job's formal mean output \(A_r\) satisfies

\[
\begin{aligned}
A_r&:=\sum_{i\in I'_r}q_{r,i}\lambda_N\left(i,H_r-\frac1{i+1}\right)\\
&\le\bigl(1+o(1)\bigr)
\frac{k_rs_r^{p-1}}{DN}.
\end{aligned}
\tag{8.9}\label{eq:9.9}
\]

If \(A_r>0\), its entire retained window lies in \(\Omega_N\), and \(i^{p-1}/s_r^{p-1}\to1\) uniformly throughout that window. The counting estimate, retained-length estimate, and this ratio all converge uniformly. Given the already fixed \(\epsilon\), choose their individual errors sufficiently small that their product is at most \(1+\epsilon\), then take the maximum of their thresholds. Since \(\sum_iq_{r,i}=k_r\), the weight in \eqref{eq:9.7} then gives

\[
A_r\le\sum_iq_{r,i}w_i.
\tag{8.10}\label{eq:9.10}
\]

If \(A_r=0\), the same inequality follows from nonnegativity, including for jobs whose retained windows meet \(\Omega_N\) only partially. Their zero output does not mean zero occupancy. Their quota copies remain in the bin permutation and displace main-position weight.

For clarity, every quota used here keeps its original \emph{global} cumulative prefix. If \(c_{r,b}\) is the total retained length of job \(r\) in bins with global number smaller than \(b\), then a supported bin has quota
\[
 t_{r,b}(V_r)=
 \left\lfloor\frac{k_r(c_{r,b}+\ell_b)}{L'_r}+V_r\right\rfloor
 -\left\lfloor\frac{k_rc_{r,b}}{L'_r}+V_r\right\rfloor.
\]
For an unsupported bin both cumulative endpoints agree, so its quota is zero. Restricting attention to \(\Omega_N\) changes neither \(c_{r,b}\) nor \(L'_r\). Recomputing either from just the local bins would in general change the allocation law. A job with a partially visible window need not contribute its full demand inside \(\Omega_N\).

\subsection{\texorpdfstring{A nonnegative mass partition and the coefficient \(1/p\)}{A nonnegative mass partition and the coefficient 1/p}}\label{sec:9.4}

For bins \(b\subseteq\Omega_N\), put

\[
a_{r,b}=\frac1{\ell_b}\sum_{i\in\mathcal B_b}
\lambda_N\left(i,H_r-\frac1{i+1}\right),
\qquad \bar w_b=\frac1{\ell_b}\sum_{i\in\mathcal B_b}w_i.
\]

The coefficient \(a_{r,b}\) is defined by this formula only when \(\mathcal B_b\subseteq I'_r\); all other such coefficients are set to zero. This condition is determined by the fixed window and does not depend on whether the realized random quota happens to vanish. On \textbf{all} outcomes of the quota marks, define the real-valued function

\[
Y(V)=\sum_{b\subseteq\Omega_N}
\left[\ell_b\bar w_b+
\sum_rt_{r,b}(V_r)(a_{r,b}-\bar w_b)\right].
\tag{8.11}\label{eq:9.11}
\]

When the quotas are valid, this is the permutation average of the inherited row jobs' output principal load plus the weights \(w_i\) at the main positions. If a bin is overloaded, \eqref{eq:9.11} is only a real-valued algebraic extension; a negative number of remaining slots is not interpreted as a geometric object.

After averaging, however, the mass assigned to main positions is nonnegative. Indeed, increase the threshold so that every relevant retained window satisfies \(L_r/L'_r\le100/99\). The input-load bound then gives

\[
Q_i:=\sum_rq_{r,i}\le\frac{100}{99}\rho_I=\frac{86}{99}<1.
\tag{8.12}\label{eq:9.12}
\]

Consequently, \eqref{eq:9.10} yields

\[
\begin{aligned}
\mathbb EY
&=\sum_rA_r+\sum_{i\in\Omega_N}(1-Q_i)w_i\\
&\le\sum_{i\in\Omega_N}Q_iw_i+
\sum_{i\in\Omega_N}(1-Q_i)w_i
=\sum_{i\in\Omega_N}w_i.
\end{aligned}
\tag{8.13}\label{eq:9.13}
\]

Every index has total mass exactly \(1\). This identity replaces any assumption about the density at which rows are generated.

The leading constant comes from an elementary integral over the \emph{narrow} support:
\[
 \frac1{DN}\int_{A^{1/p}S}^{B^{1/p}S}x^{p-1}\,dx
 =\frac{(B-A)S^p}{pDN}=\frac1p=0.8.
\]
The extra boundary length is \(O(S^{3/4})\); multiplied by the integrand size \(O(S^{p-1}/N)\), it contributes \(O(S^{-1/4})\). The integer sum over the full nominal envelope has the same limit. The actual sum, truncated at \(m\), is bounded above by that full sum; it need not have the same limit uniformly in \(m\). Hence for large scales
\[
\begin{aligned}
 \sum_{i\in\Omega_N}\frac{i^{p-1}}{DN}&\le0.8+\epsilon,\\
 \mathbb EY\le\sum_{i\in\Omega_N}w_i
 &\le(1+\epsilon)(0.8+\epsilon)\\
 &=1.001\cdot0.801=0.801801<0.81.
\end{aligned}
\tag{8.14}\label{eq:9.14}
\]
This is the numerical margin used in the verified proof. The unweighted sum has leading coefficient \(1/p\); the fixed factor \(1.001\) remains present in the weighted sum. Removing positions below \(m\) can only decrease both sums. Using a fixed-proportion enlargement instead would alter the leading integral, so those wider bands are reserved for causality and input comparisons.

\subsection{Fluctuations caused by the quota marks}\label{sec:9.5}

A job occupies \(O(S^{3/4-q})=O(S^{1/10})\) complete bins. Changing its mark \(V_r\) changes each bin quota by at most \(1\). Since \(a_{r,b},\bar w_b\in[0,h_N]\),

\[
\operatorname{osc}_{V_r}Y
\le K S^{1/10}h_N=K N^{-47/100}.
\tag{8.15}\label{eq:9.15}
\]

There are \(O(S)\) relevant jobs: summing the valid input-load bound over a slightly enlarged comparison interval containing their complete windows gives \(\sum_r k_r\le KS\), and each \(k_r\ge1\). Lemma~\ref{thm:8.1} therefore gives, for every fixed \(u>0\),

\[
\mathbb P\{Y-\mathbb EY\ge u\}
\le K e^{-\kappa_u N^{7/50}},
\tag{8.16}\label{eq:9.16}
\]

because the sum of squared bounded-difference constants is at most

\[
K S(S^{1/10}h_N)^2
=K N^{4/5+4/25-11/10}=K N^{-7/50}.
\]

No conditioning on quota success is imposed here. Such conditioning could alter both the independence of the \(V_r\) and the mean identity \eqref{eq:9.13}.

\subsection{Permutation fluctuations conditional on valid quotas}\label{sec:9.6}

Fix a valid quota outcome, and perform independent permutations in the bins. Define

\[
Z=\text{the inherited row jobs' actual output principal load}
+\sum_{\text{main positions }i}w_i.
\tag{8.17}\label{eq:9.17}
\]

Then \(\mathbb E[Z\mid V]=Y(V)\). The value associated with any label at any position lies in \([0,h_N]\), so exchanging two labels changes \(Z\) by at most \(2h_N\).

When a uniform permutation is revealed one position at a time, two possible choices of the next label can be coupled by exchanging that label with a label at an unrevealed position. The corresponding conditional expectations of the completed observable differ by at most \(2h_N\). Thus each step of the permutation-exposure martingale has conditional oscillation at most \(2h_N\). There are \(O(S)\) relevant positions, so Lemma~\ref{thm:8.1} gives

\[
\mathbb P\{Z-Y(V)\ge u\mid V\}
\le K e^{-\kappa_u N^{3/10}},
\tag{8.18}\label{eq:9.18}
\]

where \(S h_N^2=N^{4/5-11/10}=N^{-3/10}\). Dependence among positions in the same bin is incorporated into the permutation martingale; those positions are not assumed independent.

\subsection{Actual main queries and their comparison weight}\label{sec:9.7}

\textbf{Why another estimate is needed.} The frozen source observable already pays for every main position through \(w_i\). The height actually chosen by a main row can produce more load than that weight. We estimate the excess for the real sequence of height requests, including paths on which the construction later fails.

Let \(\mathcal F_j\) contain exactly the information revealed before program step \(j\). At a new main-height request whose reference lies in \([S/2,2S]\), let \(G_j\) be the selected members that lie in \(\Omega_N\), and let \(X_j\) be their load at target \(N\) after that height is read. The complete row membership and reference were selected before the request. At all other steps, including requests with reference outside \([S/2,2S]\), and after termination, set \(G_j=\varnothing\) and \(X_j=0\). No positive contribution is discarded by the reference restriction: a contributing member lies in \([S-O(S^{3/4}),B^{4/5}S+O(S^{3/4})]\), and the main-row span \(i-z=O(z^q)\) then forces \(z\in[S/2,2S]\) eventually. We count a height key only at its first request. No future permutation is inserted into \(\mathcal F_j\).

Put
\[
 \mu_j=\mathbb E[X_j\mid\mathcal F_j],\quad
 Q_t=\sum_{j<t}X_j,\quad
 U_t=\bigcup_{j<t}G_j,\quad W_t=\sum_{i\in U_t}w_i.
\]
The sets \(G_j\) are pairwise disjoint: the actual row-selection process never reuses an assigned main index. Their common height within a group is left intact. The inverse-function count in Section~\ref{sec:9.3}, applied to any one height, bounds its eligible positions by \(K(S^2N^{-p}+1)=K N^{7/20}\). Together with \eqref{eq:9.5}, this gives a group load of order \(N^{-1/5}\). The bounds used in the formal proof are
\[
\begin{gathered}
 0\le X_j\le128N^{-1/5},\qquad
 \mu_j\le\sum_{i\in G_j}w_i,\\
 \sum_{j<t}\mu_j\le W_t\le\sum_{i\in\Omega_N}w_i\le0.81.
\end{gathered}
\tag{8.19}\label{eq:9.19}
\]
Here is one explicit way to obtain 128. The row reference lies in \([S/2,2S]\); its sampled height has reciprocal scale \(s_H=1/H\le4S\). The eligible integer count is at most \(2D s_H^2N^{-p}\) eventually, including the endpoint rounding error, and the load at any eligible member is at most \(40N^{-11/20}\). Their product is at most
\[
 (2D)(16S^2)N^{-p}\,40N^{-11/20}
 =128N^{-1/5}.
\]
Thus this constant is uniform even for an empty or adaptively chosen group; its value need not be optimized. At the initial boundary all queried members are still at least \(m\). A row crossing the boundary of \(\Omega_N\) keeps its single height mark, and only its contribution is restricted to \(G_j\).

Apply the centered form of Lemma~\ref{thm:8.2} with \(b=128N^{-1/5}\), \(M=0.81\), and the random comparison weight \(W_t\). For every fixed \(t\) and \(u>0\),
\[
 \mathbb P\{Q_t-W_t\ge u\}
 \le\exp\left[-\frac{u^2}{256(0.81+u/3)}N^{1/5}\right].
\tag{8.20}\label{eq:9.20}
\]
The bound is uniform in the finite physical horizon and in \(t\). It is an unconditional estimate on actual runs. A new height remains fresh even if unrelated earlier failures are possible; after an actual stop its increment is simply zero. In particular, this argument requires neither a completed list of future main positions nor a virtual continuation of the geometry.

\subsection{The source-only bound and the remaining margin}\label{sec:9.8}

The source experiment consists only of job quota marks and bin permutations. Its bad event is contained in the union of three events: a working bin reaches occupancy \(0.90\ell_b\); the quota-averaged observable \(Y\) reaches \(0.82\); or valid quotas have \(Z_N-Y\ge0.01\).

For the first event, the input bound is \(0.86\), including both kinds of inherited jobs. Eventually \(L_r/L'_r\le100/99\), so a bin has expected occupancy at most
\[
 \frac{100}{99}\,0.86\ell_b=\frac{86}{99}\ell_b<0.88\ell_b.
\]
This leaves a fixed gap to \(0.90\ell_b\). Lemma~\ref{thm:4.2} bounds the probability for one source bin by an exponential with scale exponent \(S^{11/20}\). There are at most \(3S\) working bins, a deliberately coarse bound from their nonempty disjoint positions. Absorbing this polynomial factor gives
\[
 \mathbb P\{\exists\text{ working bin with occupancy }\ge0.90\ell_b\}
 \le e^{-c_sN^{11/25}}
\tag{8.21}\label{eq:9.21}
\]
for some \(c_s>0\) after increasing the threshold. Equality at the threshold is conservatively included in the event, although it still leaves enough main positions for the deterministic geometry.

For the second event, \eqref{eq:9.14} and the quota bound with \(u=0.01\) give an exponential with exponent \(N^{7/50}\). For the third, first fix the quota marks and apply the full permutation bound \eqref{eq:9.18}, then integrate over those marks. Invalid quotas already lie in the bin event. When none of the three events occurs,
\[
 Y<0.82,\qquad Z_N<Y+0.01<0.83.
\]
Taking a smaller common positive decay rate proves \eqref{eq:9.1}, with coefficient three.

\begin{center}
\begin{tabular}{@{}ll@{}}
\toprule
Random component & Exponent in target \(N\) \\
\midrule
Source-bin occupancy & \(11/25\) \\
Aggregate influence of job quota marks & \(7/50\) \\
Permutations with quotas fixed & \(3/10\) \\
Actual main-query compensation & \(1/5\) \\
\bottomrule
\end{tabular}
\end{center}

The first three rows bound the frozen source experiment. The fourth is the separate actual-process estimate \eqref{eq:9.20}. Section~\ref{sec:12.2} combines them using the exact margin \(0.85-0.83=0.02\); it does not assert independence between these two events.

\section{Endpoint-load estimates}\label{sec:10}

The two endpoint children share one splitting mark. Their total contribution is treated as a single nonnegative increment. The small total cut-height budget then controls the sum of its predictable means.

\begin{lemma}[Contribution of one fresh cut]
\label{thm:10.1}
Fix a parent-row height \(H\) and an endpoint width \(e\) satisfying

\[
cH^p\le e\le(1+o(1))H.
\]

Perform a fresh random split as in \eqref{eq:3.1}. Let \(X(N)\) be the sum of the fractional loads of the two children at \(N\). Uniformly for sufficiently large \(N\),

\[
\mathbb E[X(N)\mid H,e,\text{past}]\le21H,\qquad
0\le X(N)\le K N^{-11/20}.
\tag{9.1}\label{eq:10.1}
\]

\end{lemma}

\begin{proof}
A child has longer side \(H\) and shorter side \(h\). If \(N\in I(h)\), its demand is at most \(H\min I(h)\), so its load is at most

\[
(1+o(1))\frac{H N^{p-1}}D.
\tag{9.2}\label{eq:10.2}
\]

The marginal density of each child's shorter side is \(15/e\), although the two children are not independent. The eligible height interval has length \(DN^{-p}\). Summing the two mean bounds therefore gives at most

\[
(30+o(1))\frac H{eN}.
\tag{9.3}\label{eq:10.3}
\]

It suffices to consider the case in which at least one child can be eligible. Since \(h\le(2/3)e\), and eligibility implies \(h\ge1/(N+1)\),

\[
eN\ge\frac32\frac N{N+1}=\frac32-o(1).
\]

Substitution into \eqref{eq:10.3} bounds the expectation by \((20+o(1))H\le21H\). If the support does not intersect the eligible height interval, the expectation is zero and this lower bound is unnecessary.

For the maximum contribution, an eligible child satisfies \(h\ge e/3\), so \(e\le3f_B(N)=O(1/N)\). The minimum reserve implies \(H\le(e/c)^{1/p}=O(N^{-1/p})\). Substituting this into the two bounds \eqref{eq:10.2} gives the maximum \(K N^{p-1-1/p}=K N^{-11/20}\). These bounds include demand \(k=1\), partial intersection of the eligible interval with the density support, and simultaneous eligibility of both children.

\end{proof}

\begin{corollary}[Adaptive endpoint sequences]
\label{thm:10.2}
Consider any sequence of endpoint cuts performed in their actual order of creation. Suppose that, predictably before each subsequent sample, the parent-height budget is kept within

\[
\sum_r H_r\le\Pi.
\]

Then

\[
\mathbb P\{F_E(N)\ge0.01\}
\le K e^{-\kappa N^{11/20}}.
\tag{9.4}\label{eq:10.4}
\]

\end{corollary}

\begin{proof}
Combine the two children of each cut into a single nonnegative increment. By Lemma~\ref{thm:10.1}, the total predictable mean is at most \(21\Pi=0.0021\), and each increment is at most \(K N^{-11/20}\). Apply Lemma~\ref{thm:8.2} with \(u=0.01-21\Pi>0\). If the actual process stops early, pad the remaining increments with zeros. The reservation ledger retains children whose rows have already started, so this sum of increments controls the load of all relevant endpoint jobs, rather than only their current waiting area.

\end{proof}

\section{Revelation order in the recursive construction}\label{sec:11}

\textbf{Beginning Step 5.} The preceding estimates are useful only if their random marks are still fresh at the moment they are used. We now specify that moment. The physical cursor and the discrete program-step count are different clocks, and both will be needed. Deterministic transitions are included in that step count even when no random coordinate is read.

\subsection{Random coordinates and physical operations}\label{sec:11.1}

A countable family of independent uniform random coordinates may be assigned in advance to all objects that can be created, distinguishing objects by the name of their parent and their creation index. Assigning an unused coordinate does not reveal its value. The construction depends only on coordinates that have been revealed and follows the rules below.

\begin{enumerate}
\item The geometry, window, and demand of each inherited row job are registered when its parent row is activated. Its own quota variable \(V_r\) remains unrevealed at this time.
\item For a bin starting at \(s_b\), its allocation data are revealed at physical time \(\lfloor\gamma s_b\rfloor\), where \(\gamma=9/10\). A job's shared variable \(V_r\) is first revealed at the revelation time of its leftmost complete bin; all subsequent quotas for that job use the same value. Given the quotas, the uniform permutations of the bins are revealed independently.
\item If the occupancy of a bin exceeds \(0.90\), a bad certificate is recorded, and no unavailable label permutation is required. This advance notification neither changes earlier geometric decisions nor immediately stops physical processing.
\item All allocations of a job are verified no later than \(\min I_r\). The verification checks the input-load certificates, occupancy certificates, and resulting allocations of the complete bins used by the job. If a required certificate fails, the process stops at this time. If the allocation is valid, its geometry is nevertheless activated only at its first \textbf{actually assigned index}. A job without a valid allocation cannot be deferred until after its window has expired.
\item At every physical integer \(N\), the inequalities \(F_P(N)\le0.85\) and \(F_E(N)\le0.01\) are checked, whether or not the realized label at that index happens to fit. When a main row needs to read the next bin, that bin's certificates are also checked.
\item A main row selects its complete set of members and its reference before revealing its own height variable. Every row determines its height and terminal width before revealing its endpoint-splitting variable. Before the next split, the accumulated sum of cut-parent heights plus the new parent height is checked against \(\Pi\); if it would exceed \(\Pi\), the process stops before sampling. The increment \(2H+e\) in the endpoint-semiperimeter account is likewise determined before sampling and may be checked in the same manner.
\end{enumerate}

A certificate is simply a specified inequality to be checked. There is no rejection sampling, resampling, or selection of current marks according to future success or failure. Every possible reason for stopping is included in the first-failure estimate of Section \ref{sec:12}.

\begin{lemma}[Timely availability of data]
\label{thm:11.1}
For all sufficiently large initial indices \(m\), the preceding rules can be executed in the stated order at every step. Any inherited row job whose window contains indices near scale \(s\) is created before physical time \((3/4)s\). The data of all complete bins intersecting that window have been revealed before \(\min I_r\).

\end{lemma}

\begin{proof}
If the job is a principal gap, its source index \(i\) satisfies \(g_i\asymp i^{-p}\), whereas its future window has scale \(s\asymp g_i^{-1}\). Thus \(i=O(s^{1/p})\). The parent row creating this gap is activated no later than physical time \(i\), so the creation time is \(O(s^{4/5})<3s/4\).

If the job is an endpoint child, let \(H\) be the parent height. The child's short side satisfies \(h\le(2/3)(1+\eta)H\). The parent's first index is \((1+o(1))/H\), and any admissible index of the child is \((1+o(1))/h\). Consequently, the ratio of the creation time to any admissible index of the child is at most

\[
\frac23(1+\eta)+o(1)<\frac34.
\tag{10.1}\label{eq:11.1}
\]

The same asymptotic relation between height and the reciprocal of the first index holds for main rows.

For a bin near \(s\), the entire windows of all intersecting jobs lie in \(s+O(s^{3/4})\). Their creation times are at most \((3/4)(s+O(s^{3/4}))\), strictly earlier than \(\gamma s\). Thus the input pool is complete when the bin data are revealed.

Conversely, the start of the last complete bin of a job is at most \(s+Ks^{3/4}\), and its earliest admissible index is at least \(s-Ks^{3/4}\). Since \(\gamma<1\),

\[
\gamma(s+Ks^{3/4})<s-Ks^{3/4}
\]

holds eventually. Hence the entire allocation is available before the verification deadline.

A job's variable \(V_r\) may be revealed earlier than some later bin that it affects, but its earliest revelation time is still \(\gamma s-O(s^{3/4})\). This is strictly later than the formation of the relevant job geometries. Accordingly, the variable is not disclosed while the input pool at that scale is being formed.

\end{proof}

At initialization, a bin whose nominal revelation time is less than \(m\) is processed at the initial time. Such bins have starts below \(m/\gamma+O(1)\). Every new endpoint window has indices at least \((4/3-o(1))m\), and every new principal-gap window has indices at least a fixed positive constant times \(m^p\). Thus these initial bins have empty inherited pools and consist entirely of main positions. They cannot conflict with subsequently created jobs. More generally, the entire window of a new job lies strictly after its creation time, so no job requires an index below \(m\).

\subsection{Source bands and the actual freezing time}\label{sec:11.2}

Fix \(N\) and put \(S=N^{4/5}\). Select exactly those bins whose real half-open spans meet the narrow closed interval
\[
 [S-4S^{3/4},\ B^{4/5}S+4S^{3/4}].
\]
Equivalently their starts satisfy
\[
 s_b\le B^{4/5}S+4S^{3/4},\qquad S-4S^{3/4}<s_{b+1}.
\]
Their integer positions form \(\Omega_N\). Specifying real spans avoids a one-bin ambiguity when a noninteger interval endpoint lies between two integer positions.
Their positions lie in the analogous envelope with coefficient six. All complete windows of their input jobs fit in the comparison envelope with coefficient eight. Only bins of the construction, whose positions are at least \(m\), are used. For sufficiently large \(S\), the relevant narrow envelopes lie inside the fixed-proportion bands
\[
 J_0=[0.995S,1.085S],\quad
 J_1=[0.99S,1.09S],\quad J_2=[0.98S,1.10S].
\]
An additional window intersecting a comparison band has span \(O(S^{3/4})=o(S)\) and fits in the next wider band. These wider bands supply causal separation only.

The physical freezing threshold is
\[
 T_{m,N}=\max\{m,\lfloor0.85S\rfloor\}.
\tag{10.2}\label{eq:11.2}
\]
Let \(\tau\) be the \emph{first} program step whose actual cursor reaches this threshold, if such a step exists. On paths that never reach it there is no freezing case. If the threshold is \(m\), take the initial history \(\tau=0\), before any initialization queries.

By Lemma~\ref{thm:11.1}, every job supplying the relevant working or comparison positions is created before \(0.825S\). The earliest request of any of its shared quota marks is at least \(0.882S-o(S)\). Consequently,
\[
 0.825S<0.85S<0.882S-o(S).
\tag{10.3}\label{eq:11.3}
\]
The same separation puts all working-bin permutation requests after the freeze. It controls the earliest bin in each job's \emph{whole window of action}, including bins outside the working region. Checking only that one working bin is unread would not establish that its shared job mark is unread.

When \(m<\lfloor0.85S\rfloor\), every query before \(\tau\) has cursor below the nominal freezing point, so \eqref{eq:11.3} applies directly. When \(T_{m,N}=m\), the history is empty and no marks have been read. In this initial case, any future supplier would have to be created before \(0.825S<m\), which is impossible for large \(S\); the input pool is empty. This includes the bounded rounding discrepancy in the floor by using the strict proportional margin. Actual positions below \(m\) remain absent.

The pool at \(\tau\) is complete: a later-created job cannot supply these bins, since that would violate the same early-creation estimate. Registered widths, heights, demands, and job identifiers are retained permanently. Thus the pool's geometric data are fixed at the freeze and unchanged later, while its quota and working-permutation coordinates remain fresh. A positive actual principal load at \(N\) implies that this freeze has been reached, with one additional timing check. A contributing source member \(i\) satisfies \(i\ge S-2S^{3/4}\), but it may be a future reserved member rather than the parent's activation index. If its row is committed at cursor \(c\), the actual row-span bound gives \(i\le c+3c^{3/4}\). If \(c\le0.85S\), these inequalities would give
\[
 S-2S^{3/4}\le0.85S+3(0.85S)^{3/4},
\]
which is false for sufficiently large \(S\). Thus the cursor has already passed the nominal freeze, unless the freeze was the initial history. This argument uses the row's activation span, not an assumption that a future member has already been physically processed.

\subsection{Applying a fixed-pool estimate at a random history}\label{sec:11.3}

\textbf{The measurability issue.} The pool contains random real widths and heights, and \(\tau\) is random. We therefore do not condition on an individual real-valued history as if it necessarily had positive probability.

Jobs have identifiers from a countable set. Partition the paths that reach the freeze by pairs
\[
 c=(t,J),\qquad
 A_c=\{\tau=t,\ \text{the source identifier set at }t\text{ is }J\},
\]
where \(t\in\mathbb N\) and \(J\) is a finite set of identifiers. There are countably many cases, and the events \(A_c\) are pairwise disjoint. Both the first-reaching condition and the identifier set are determined by the history at \(t\). The continuous widths and heights are not discretized; they remain measurable functions of that history.

Write \(H_t\) for the revealed history and \(R_c\) for the entire finite block of source coordinates required in case \(c\). To justify the joint law, fix the finite set of unread keys for this case. Changing all coordinates at these keys leaves the preceding history and its unreadness unchanged: an induction along the recorded program steps shows that the decisions and replies so far use only coordinates outside this set. Conversely, replacing this block by fixed values does not change whether such a history is realized without requesting it. Therefore any measurable event specifying that history together with unreadness depends only on the outside coordinates. In the original product space, the whole unread block is independent of the outside coordinates. Applying independence to a measurable history event and a measurable block event proves the product identity on rectangles; uniqueness of finite product measures gives the full joint law. This establishes independence of the entire block, rather than only one-coordinate marginal distributions. Applying this property to \(A_c\), using the unreadness just proved, gives
\[
 (\mathbb P|_{A_c})\circ(H_t,R_c)^{-1}
 =\bigl((\mathbb P|_{A_c})\circ H_t^{-1}\bigr)\otimes\nu_c,
\]
where \(\nu_c\) is the full joint law of the remaining quota and permutation coordinates. The equality concerns the entire block, including dependencies later introduced by using one job mark across several bins.

Let \(G(h)\) mean that the actual frozen registry has principal load at most \(0.85\) and endpoint load at most \(0.01\) at every comparison index. In the joint history/source space define the bad section by
\[
 B_c=\{(h,r):G(h)\ \text{and}\
       [\,\operatorname{Overflow}(h,r)\ \text{or}\ Z_N(h,r)\ge0.83\,]\}.
\]
Here overflow means occupancy at least \(0.90\ell_b\) in some working bin. All operations involved are measurable: finite registry selection, real arithmetic and floors, finite sums, and finite permutation decoding. For almost every \emph{actual} history, the construction invariants provide the geometric hypotheses of Section~\ref{sec:9}. If \(G(h)\) fails, this section is empty. If \(G(h)\) holds, the fixed-pool source bound applies uniformly. There is no assertion that an arbitrary, unreachable assignment of history data satisfies those invariants.

Fubini's theorem therefore yields
\[
 \mathbb P\bigl(A_c\cap\{(H_t,R_c)\in B_c\}\bigr)
 \le\mathbb P(A_c)\,3e^{-\kappa_sN^{7/50}}.
\]
Summing over the disjoint countable cases gives a single frozen bad event \(\mathcal S_N\) with
\[
 \mathbb P(\mathcal S_N)\le3e^{-\kappa_sN^{7/50}},
 \qquad \sum_c\mathbb P(A_c)\le1.
\tag{10.4}\label{eq:11.4}
\]
The history probabilities must be kept in this sum. Replacing each weighted bound by the same unweighted constant and summing would lose the estimate. There is no factor counting queries, possible pools, or real-valued geometric choices.

\textbf{Why the literal source experiment has the same law.} The actual program may number a bin's suppliers differently from the common source-pool list. Fix the quotas first. Delete entries with zero copies in this bin, then reorder the remaining job copies by their actual identifiers. This gives a finite bijection between the two labeled-copy lists, including the remaining main copies. Composing a uniform bin permutation with this bijection preserves its uniform law and identifies every physical label, not merely its one-position marginal. Taking products handles all bins. The correcting bijection may depend on the quotas, so this equality is proved in each quota section and then integrated over the quota marks. Empty lists and zero quotas have unique empty contributions and cause no exception.

Finally, the splitting and height uniforms satisfy their support restrictions simultaneously outside one null set, since there are countably many available coordinates. This permits all subsequent comparisons to hold on a common full-probability set across all actual histories and program times. No conditional distribution is altered by imposing eventual successful execution.

\section{Actual obstructions and finite-prefix existence}\label{sec:12}

This section completes Step 5. The purpose of a least-obstruction argument is to make the source input good using strictly smaller \emph{indices}. The order in which a cached certificate is eventually used is a different order and is not used to define the probability events.

\subsection{Events indexed by the detected integer}\label{sec:12.1}

Fix \(m\) and a finite terminal index \(L\). Retain the state after either success or failure, so it is defined at every later program step. Put
\[
 P_N=\{\exists t:F_P(N,t)>0.85\},\qquad
 E_N=\{\exists t:F_E(N,t)\ge0.01\}.
\tag{11.1}\label{eq:12.1}
\]
Let \(B_b\) be the event that, at some actual program step whose physical cursor is at least the revelation threshold of bin \(b\), its actual input check is good and its actual quota occupancy is at least \(0.90\ell_b\). A bin with bad input is charged to the detected load index that caused its input failure; it is not counted as a good-input occupancy event.

For \(N\ge m\), define the reduced principal event
\[
 D_N=\left(P_N\cap\bigcap_{m\le i<N}P_i^{\,c}\right)
       \setminus\left(\bigcup_{i\ge m}E_i\ \cup\ \bigcup_bB_b\right).
\tag{11.2}\label{eq:12.2}
\]
Thus \(N\) is the least integer at which principal overload ever occurs, after the endpoint and good-input bin events have been removed. This definition involves the whole path, but it will only be used for deterministic event inclusions. We do not condition the random marks on \(D_N\).

Every actual load or good-input bin obstruction lies in
\[
 \bigcup_{i\ge m}E_i\ \cup\ \bigcup_bB_b\ \cup\ \bigcup_{N\ge m}D_N.
\tag{11.3}\label{eq:12.3}
\]
Indeed, if neither of the first two unions occurs but some principal overload does, the nonempty set of its indices at least \(m\) has a least element. This elementary well-ordering step requires no independence.

\begin{lemma}[Good comparison input and good used caches]\label{thm:12.1}
For sufficiently large \(m\), on \(D_N\) every comparison position for target \(N\) has good input at every program step. Every already cached working bin has a good certificate. These claims impose no additional assumption on \(m\) relative to the left endpoint of the source region.
\end{lemma}
\begin{proof}
All comparison positions are at most \(2N^{4/5}<N\) for large \(N\). If such an index \(i\) is at least \(m\), the definition of \(D_N\) gives \(F_P(i,t)\le0.85\) and \(F_E(i,t)<0.01\) for every \(t\). If \(i<m\), its load is zero: any registered job was created at a cursor at least \(m\), whereas membership of \(i\) in its window would require creation no later than \(3i/4<m\).

A cached bin need not be good merely because it is cached. Inspect a possible bad cached certificate at its original creation time. A bad-input certificate contains a tested index in the comparison region; its principal or endpoint violation contradicts the preceding bounds, using permanence of the registered loads if the certificate is consulted later. If the input was good and the recorded failure was an occupancy failure, the path belongs to \(B_b\), also excluded by \(D_N\). These are all ways to store a bad certificate, so none can occur in a cached working bin on this event.

At initialization the cache is empty, even though the main frontier equals \(m\). The needed state invariant is therefore a disjunction: either the registry is empty and the frontier is \(m\), or the required frontier has already been cached. The first main selection establishes the second alternative. In particular, the initial closed interval \([m,m]\) is not silently assumed to have been read.
\end{proof}

\subsection{The actual load comparison, one step at a time}\label{sec:12.2}

On \(D_N\), positive principal load ensures an actual first freeze \(\tau\), by Section~\ref{sec:11.2}. The preceding lemma makes its input good. Let \(Z_N\) now denote the \emph{literal} frozen source cost using that history and the actual source coordinates, interpreted as in Section~\ref{sec:11.3}. Recall the actual quantities \(Q_t,W_t\) from Section~\ref{sec:9.7}.

\begin{lemma}[One-sided comparison with the frozen source]\label{thm:12.2}
Outside a single null set, for all sufficiently large \(m\), on \(D_N\) and at every program time \(t\ge\tau\),
\[
 F_P(N,t)\le Z_N+Q_t-W_t.
\tag{11.4}\label{eq:12.4}
\]
\end{lemma}
\begin{proof}
We separate the comparison into three charges, so that no unprocessed label is treated as a completed row.

\textbf{First charge: committed main rows.}
Let \(M_t\) be the principal load created by main rows already entered in the permanent ledger. Between a height query and the subsequent split/commit operation, let \(\Delta_t\ge0\) be the load of the pending main row; it is zero at other phases. An induction on the actual operations proves
\[
 M_t+\Delta_t\le Q_t.
\]
Initially all three quantities are zero. A new height request increases \(Q_t\) by the selected row's load and allows exactly that load into \(\Delta_t\). Committing the row moves its load from \(\Delta_t\) to \(M_t\), without requesting its height again. Unrelated operations leave the charge unchanged; an intervening failure can discard pending load, which preserves the inequality. No row is charged twice. The source-support result permits us to restrict all these loads to \(\Omega_N\), since an actual contributing member is at least \(m\) and lies in that region.

\textbf{Second charge: inherited rows and used main weight.}
Let \(I_t\) be the principal load created by already committed inherited rows. Their contributing members have actual inherited labels in cached working bins, and the parent job, with the same height, belongs to the frozen source pool. The main indices in \(U_t\) have actual main labels, are pairwise distinct, and are disjoint from these inherited positions. Sum the source cost only over positions in already cached working bins, assigning zero to an uncached position; call this sum \(C_t\). The two disjoint classes of charged positions give
\[
 I_t+W_t\le C_t.
\]
This is a pointwise inequality: any unused position contributes a nonnegative extra term. A preliminary default ``main'' label at an unread bin is not evidence that its final label will be main. Such a position is not charged.

\textbf{Third charge: cached cost versus the whole frozen experiment.}
By Lemma~\ref{thm:12.1}, every cached working bin on \(D_N\) has a good certificate. Its true labels agree with those obtained from the frozen records and the actual shared marks. The label-recoding argument in Section~\ref{sec:11.3} is used only to identify the probability law of this same literal experiment. In comparing paths, we use the literal physical labels. Thus each cached-position cost agrees with its frozen cost, and nonnegativity of all other position costs gives
\[
 C_t\le Z_N.
\]
The permanent principal ledger splits exactly as \(F_P(N,t)=M_t+I_t\): every registered principal gap has a unique source tile and a parent row of one of these two types. Adding the first two charges and dropping the nonnegative pending term yields
\[
 F_P(N,t)+W_t\le Q_t+C_t\le Q_t+Z_N,
\]
which is \eqref{eq:12.4}. This proves an inequality for the actual process, including stopped paths; no virtual rows or future geometric completion are introduced.
\end{proof}

\textbf{What the comparison achieves.} The source experiment pays once for every inherited contribution and every unit of used main weight. What is left is precisely the actual main-query excess \(Q_t-W_t\), for which the unconditional estimate \eqref{eq:9.20} has already been proved.

\subsection{Persistent detection and the principal tail}\label{sec:12.3}

\begin{lemma}[Persistent detection from fixed-time bounds]\label{thm:12.3}
Let \((A_t)_{t\in\mathbb N}\) be measurable events with \(\mathbb P(A_t)\le b\) for all \(t\). Then
\[
 \mathbb P\left(\bigcup_{s\ge0}\bigcap_{t\ge s}A_t\right)\le b.
\]
\end{lemma}
\begin{proof}
Put \(C_s=\bigcap_{t\ge s}A_t\). The events \(C_s\) increase with \(s\), and \(C_s\subseteq A_s\), so each has probability at most \(b\). Continuity of measure from below gives \(\mathbb P(\bigcup_sC_s)=\lim_s\mathbb P(C_s)\le b\). This is not a union bound over \(A_t\). An event that occurs once need not satisfy the conclusion; the hypothesis used here is that it occurs at every sufficiently late time.
\end{proof}

\begin{lemma}[Actual reduced principal-load tail]\label{thm:12.4}
There are \(m_1\ge1\) and \(\kappa>0\), independent of the terminal index \(L\), such that for every \(m\ge m_1\), every finite \(L\), and every \(N\ge m\),
\[
 \mathbb P(D_N)\le4e^{-\kappa N^{7/50}}.
\tag{11.5}\label{eq:12.5}
\]
\end{lemma}
\begin{proof}
Split \(D_N\) according to the frozen source bad event \(\mathcal S_N\). Its probability is bounded by \(3e^{-\kappa_sN^{7/50}}\), with no conditioning on \(D_N\). On \(D_N\setminus\mathcal S_N\), the actual first freeze has good input and its literal cost satisfies \(Z_N<0.83\). Some time has \(F_P(N,t)>0.85\). Since the ledger is permanent, this strict inequality persists at all subsequent times, including after termination. Lemma~\ref{thm:12.2} consequently gives
\[
 Q_t-W_t\ge0.02\qquad\text{for all sufficiently large }t.
\]
Apply Lemma~\ref{thm:12.3} to \(A_t=\{Q_t-W_t\ge0.02\}\). By \eqref{eq:9.20}, this persistent event has probability at most
\[
 e^{-\kappa_hN^{1/5}},\qquad
 \kappa_h=\frac{(1/50)^2}{256(0.81+(1/50)/3)}>0.
\]
Take \(\kappa=\min\{\kappa_s,\kappa_h\}\). For \(N\ge1\), \(N^{1/5}\ge N^{7/50}\), and hence
\[
 \mathbb P(D_N)
 \le3e^{-\kappa_sN^{7/50}}+e^{-\kappa_hN^{1/5}}
 \le4e^{-\kappa N^{7/50}}.
\]
All supporting thresholds are uniform over actual finite histories. Taking their maximum gives the stated \(m_1\). Neither a query-count factor nor conditioning on successful future processing appears.
\end{proof}

\subsection{Endpoint and bin events}\label{sec:12.4}

For endpoints, use the actual sequence of splitting queries. Geometry and endpoint width are fixed before a split; its two children are one increment. The cut-height budget is checked before the mark is sampled. Stop and pad with zero increments on paths that fail earlier. Section~\ref{sec:10} and the geometric invariant give a uniform fixed-time upper tail. The permanent endpoint ledger is nondecreasing, so the events that its load has reached \(0.01\) by time \(t\) increase. Passing to their union preserves the same uniform bound. Thus, for some \(\kappa_E>0\),
\[
 \mathbb P(E_N)\le e^{-\kappa_E N^{11/20}}.
\tag{11.6}\label{eq:12.6}
\]
No sequence of hypothetical endpoint cuts is added. A budget failure cannot be the first physical failure because Section~\ref{sec:7} improves both budgets strictly below \(\Pi/2\).

For bin \(b\), freeze its input near physical cursor \(\lfloor0.8s_b\rfloor\), or at the initial history if that nominal time is earlier than \(m\). Relevant jobs are created by \(0.75s_b+o(s_b)\); their first quota requests occur no earlier than \(0.9s_b-o(s_b)\). This leaves the same kind of strict unreadness interval as before. Restrict the bad section to actual histories with good input. Shared quota marks are independent across jobs at this past freeze, so the occupancy estimate applies. The countable first-history argument retains its history weights, as in Section~\ref{sec:11.3}. It yields
\[
 \mathbb P(B_b)\le e^{-\kappa_Bs_b^{11/20}},\qquad\kappa_B>0,
\tag{11.7}\label{eq:12.7}
\]
uniformly in the terminal index. A bad-input cache is charged to its original tested load index, even when it is read physically much later. It is never treated as a fresh good-input experiment at the later time.

\subsection{Termination, the union bound, and finite existence}\label{sec:12.5}

\begin{lemma}[Termination of the finite-horizon program]\label{thm:12.5}
For every positive initial index \(m\), every finite terminal index
\(L\), and every realization of the random coordinates, the program
reaches either success or a recorded failure after finitely many
steps.  No assumption that its geometric or probabilistic checks
succeed is needed for this assertion.
\end{lemma}

\begin{proof}
We count all program transitions, including deterministic steps
that reveal no new random coordinate.  Let \(N\) be the physical
cursor, and let \(\mathcal C\) be the finite set of bins whose
results have already been cached, including bad results.
The set of bins due at \(N\) is
\[
 \mathcal D_m(N)
 =\{b\ge0:\lfloor\gamma s_b\rfloor\le N\},
 \qquad \gamma=\frac9{10}.
\]
It is finite for every \(N\).  Indeed, \(s_b\ge m+b\ge b\), and
the defining inequality gives
\(\gamma s_b<N+1\), so \(b<2N+2\).
Put
\[
 u=\bigl|\mathcal D_m(N)\setminus\mathcal C\bigr|.
\]

Partition the control phases into blocks, with the integer values
specified in the following table.  A stage counter resolves the
order within blocks that contain several phases.
\begin{center}\small
\begin{tabular}{@{}lcc@{}}
\toprule
Phase & Block \(b_{\rm ctl}\) & Stage \(d_{\rm ctl}\)\\
\midrule
Prepare the current cursor & 5 & 0\\
Select the next due bin & 4 & 3\\
Read its quota marks & 4 & 2\\
Read its permutation and cache the result & 4 & 1\\
Check current fractional loads & 3 & 0\\
Check the finite list of registered allocations & 2 & 0\\
Select the current row & 1 & 5\\
Read a main-row height, when required & 1 & 4\\
Check the endpoint budgets & 1 & 3\\
Read the split mark and commit the row & 1 & 2\\
Advance the physical cursor & 1 & 1\\
Success or recorded failure & 0 & 0\\
\bottomrule
\end{tabular}
\end{center}
Let \(q_{\rm ctl}\) be the number of entries remaining in the quota list
during the quota phase, or in the allocation list during the
allocation phase, and put \(q_{\rm ctl}=0\) in the other phases.
Each of these lists is finite when created: a bin's quota list
comes from its finite registered pool, and the allocation list
comes from the finite registry.

Use the five-component progress value
\[
 \mathcal R
 =\bigl(\max\{L+1-N,0\},\,b_{\rm ctl},\,u,\,d_{\rm ctl},\,q_{\rm ctl}\bigr)
 \in\mathbb N^5
\]
with the lexicographic order, earlier components taking priority.
We verify that every step from a nonterminal reachable state
strictly decreases \(\mathcal R\).

\begin{enumerate}
\item Preparation either detects \(N>L\) and ends successfully,
or moves from block \(5\) to block \(4\).

\item Selecting a due bin either finds none and moves from block
\(4\) to block \(3\), or selects an uncached due bin and moves from
stage \(3\) to stage \(2\).  In the latter case the new finite quota
list may have arbitrary length, since its length is a later
component of the progress value.

\item A nonempty quota list loses its first entry at the next
step.  Once it is empty, a successful quota check moves from stage
\(2\) to stage \(1\).  A bad input or quota result is instead
cached immediately.  Caching inserts the current bin, which was
due and uncached, into \(\mathcal C\); hence \(u\) strictly
decreases, even though the program returns to stage \(3\).

\item Completing a bin permutation also caches that bin and
strictly decreases \(u\).  This conclusion holds whether the
stored result is usable or bad.  Cached results are never removed.

\item The fractional-load check either fails or moves from block
\(3\) to block \(2\), creating a finite allocation list.  Each
allocation check either fails or removes its first list entry.
When the list becomes empty, control moves from block \(2\) to
block \(1\).

\item Within block \(1\), every nonfailing transition strictly
decreases the stage counter: row selection, optional height
sampling, budget checks, split and commit, and then advancement.
Already reserved tiles may skip intermediate stages.  A commit
may create new jobs, but it does not return to allocation checks
at the same cursor.

\item A successful advancement replaces \(N\) by \(N+1\).
An advancement phase in a nonterminal reachable state has
\(N\le L\): only preparation or a terminal state may occur beyond
the horizon.  Thus the first component strictly decreases.
The other four components may now reset.  An unsuccessful
advancement is a terminal failure.
\end{enumerate}

Every failure decreases the block to zero.  The due-and-uncached
property used in the caching cases follows from the rule selecting
the next bin, and is preserved while its quota and permutation
phases run.  These are control properties of the actual program;
they require no assumption about eventual success.

The lexicographic order on \(\mathbb N^5\) is well founded.
An infinite nonterminal run would therefore give an impossible
infinite strictly descending sequence of progress values.
The program consequently reaches a terminal state in finite time.
Terminal states are thereafter left unchanged.
\end{proof}

The termination proof does not need a uniform bound on the number
of program steps or an a priori bound on the depth of endpoint
trees.  A shared quota mark is sampled once and reused on any
later request for the same mark.  The bound
\(L+O(L^{3/4})\) applies to the indices of tiles reserved by rows
that have already started, as used in the area estimate.
It is not a bound on all bins revealed in advance: a bin is due
whenever \(\lfloor\gamma s_b\rfloor\le L\), which allows starts
near \(L/\gamma\).  Their finiteness is supplied by the due-bin
bound in the proof.

If none of the events on the right of \eqref{eq:12.3} occurs, all load and bin certificates needed by the process are good. The state and area invariant of Section~\ref{sec:7} then excludes the first geometric or budget failure. Termination therefore implies success, with cursor strictly greater than \(L\); the extracted packing includes \(R_L\). Rows may still contain future reserved tiles, which can simply be omitted when extracting the finite prefix.

Since \(\ell_b\ge1\) for \(m\ge1\), the bin starts satisfy \(s_b\ge m+b\). Countable subadditivity, \eqref{eq:12.5}--\eqref{eq:12.7}, and this monotone comparison give
\[
\begin{aligned}
 \mathbb P\{\text{an obstruction}\}
 &\le\sum_{n=m}^{\infty}e^{-\kappa_E n^{11/20}}
       +\sum_{n=m}^{\infty}e^{-\kappa_B n^{11/20}}\\
 &\quad+4\sum_{n=m}^{\infty}e^{-\kappa n^{7/50}}.
\end{aligned}
\tag{11.8}\label{eq:12.8}
\]
There is no need to shift the lower limit by a fixed proportional factor: the events are indexed directly by actual tested indices \(n\ge m\) and by bins with \(s_b\ge m+b\).

For any \(a,c>0\), the series \(\sum_{n\ge1}e^{-cn^a}\) converges. For example, eventually \(cn^a\ge2\log n\), so its terms are at most \(n^{-2}\). Its tails tend to zero. Increase a single initial threshold so that the right-hand side of \eqref{eq:12.8} is less than \(1/2\).

\begin{theorem}[Uniform finite-prefix existence]\label{thm:12.6}
There exists \(m_0\ge1\) such that, for every \(m\ge m_0\) and every integer \(L\ge m\), the rectangles \(R_m,\ldots,R_L\) have an interior-disjoint packing in \(Q_m\).
\end{theorem}
\begin{proof}
For these \(m,L\), the complement of the obstruction event has positive probability, hence is nonempty. On any sample in that complement, the actual construction terminates successfully by the preceding deterministic argument. Its placed rectangles give the desired finite packing. The bound is uniform in \(L\), so the same \(m_0\) works for every finite prefix.
\end{proof}

The successful sample may depend on \(L\). No infinite successful random trajectory is asserted here; that is unnecessary for the deterministic compactness argument that follows.

\section{A compactness limit and the area identity}\label{sec:13}

\textbf{Step 6.} The probability argument provides existence separately for every finite terminal index. The following deterministic argument combines those existences and then uses the exact area sum.

\begin{lemma}[Compactness of finite prefixes]
\label{thm:13.1}
If, for every \(L\ge m\), the rectangles \(R_m,\ldots,R_L\) can be packed into the same closed square \(Q_m\), then the entire infinite tail can be packed into \(Q_m\).

\end{lemma}

\begin{proof}
Write \(s=m^{-1/2}>0\). A placement coordinate for index \(i\) is its lower-left corner and orientation,
\[
 (x_i,y_i,\sigma_i)\in X_i=[0,s]^2\times\{0,1\}.
\]
The orientation set has the discrete topology and is finite. Each \(X_i\) is compact, so the countable product \(X=\prod_{i\ge m}X_i\) is compact. For each \(L\ge m\), let \(C_L\subseteq X\) impose containment and interior disjointness only on indices \(m,\ldots,L\). A finite packing gives a member of \(C_L\) by filling the unused coordinates arbitrarily. Thus \(C_L\) is nonempty, and \(C_{L+1}\subseteq C_L\).

Given the orientation, containment is expressed by

\[
0\le x_i,\quad 0\le y_i,\quad
x_i+a_i\le m^{-1/2},\quad y_i+b_i\le m^{-1/2},
\]

where \((a_i,b_i)\) lists the side lengths in that orientation. These conditions are closed. Two fixed axis-parallel rectangles of positive area have disjoint interiors if and only if at least one of the following four inequalities holds:

\[
x_i+a_i\le x_j,\quad x_j+a_j\le x_i,\quad
y_i+b_i\le y_j,\quad y_j+b_j\le y_i.
\tag{12.1}\label{eq:13.1}
\]

This condition is a finite union of four closed sets and is therefore closed. Orientation is discrete, so the same closedness holds when orientation is included as a coordinate. Only finitely many constraints occur in \(C_L\); hence \(C_L\) is closed in \(X\). The nested nonempty closed subsets \(C_L\) of the compact space \(X\) have a common point. Its coordinates satisfy every single-tile and every pairwise constraint, since each constraint occurs in some finite prefix. This common point is the desired packing. Non-strict separation inequalities are essential: they preserve boundary contact in the limit.

\end{proof}

Applying Lemma~\ref{thm:13.1} to the finite-horizon existence result of Section \ref{sec:12} proves \eqref{eq:1.1}. Finally,

\[
\sum_{n=m}^{\infty}|R_n|
=\sum_{n=m}^{\infty}\left(\frac1n-\frac1{n+1}\right)
=\frac1m=|Q_m|.
\tag{12.2}\label{eq:13.2}
\]

The rectangle interiors are pairwise disjoint. Each boundary has area zero, and a countable union of these boundaries still has area zero. Apply countable additivity first to the open interiors, which are actually disjoint. Adding back the countable union of zero-area boundaries leaves the area unchanged. Thus the area of the union of the closed rectangles equals the sum above. Countable additivity is not being applied to closed rectangles as though their boundaries were disjoint. This union is measurable and contained in \(Q_m\); its complement therefore has area zero, proving \eqref{eq:1.2}. \textbf{This completes the proof of Theorem 1.1.}

\section{\texorpdfstring{Selection of a single threshold \(m_0\)}{Selection of a single threshold m0}}\label{sec:14}

The requirements implicit in ``sufficiently large'' are imposed in the following order. This also shows that one threshold can be chosen uniformly over all finite terminal indices.

\begin{enumerate}
\item Fix the numerical values in \eqref{eq:1.3}--\eqref{eq:1.7}, \(c_0=1/4\), and \(\gamma=9/10\), and fix \(a_0\) satisfying \eqref{eq:7.3}.
\item Use the uniform estimates for windows and main-row spans to fix the coefficient \(K_0\) defining the interior height interval. Then fix the source and comparison buffers and the mean-envelope margin \(\epsilon=10^{-3}\). These coefficients do not depend on the random history.
\item Choose a height threshold ensuring that admissible windows are nonempty, \(a_{\min}\ge(1-\eta)H\), \(a_{\max}\le H-cH^p\), and \(e\le(1+\eta)H\), and that all short-side orientation conditions and causal separations hold. Main-row heights are at most \(3/(2m)\), inherited-row heights are at most \(f_B(m)\), and each endpoint child has short side less than three quarters of its parent-row height. A sufficiently large \(m\) therefore keeps every required height below this threshold. Termination itself uses the program rank of Lemma~\ref{thm:12.5}, not convergence of a sequence of heights.
\item Increase \(m\) so that complete bins exist, \(L'_r/L_r\) is sufficiently close to one, the main-height interval is nonempty, and the main-strip fitting and group-size estimates hold. Also ensure that the endpoint budget in Section \ref{sec:7} is strictly below \(\Pi/2\), and that the area lower bound is strictly stronger than the inductive lower bound.
\item Use the fixed margins between the probabilistic thresholds to fix their upper-tail constants \(K,\kappa\). Increase \(m\) further so that all uniform errors are smaller than their allocated margins and the tail sum in \eqref{eq:12.8} is less than \(1/2\).
\end{enumerate}

These are finitely many threshold conditions together with one condition on the tail of a convergent series. Taking their maximum gives a single integer \(m_0\).

\noindent\textbf{Explicit numerical calculation.}
We now quantify the five requirements above. A deliberately generous
sufficient choice is
\[
                         \boxed{m_0=10^{1000}}.
\]
The following estimates verify this choice uniformly over all finite
terminal indices; no optimization of the threshold is intended.

Keep all the numerical parameters fixed in Section 1 and choose
\[
 a_0=\frac1{10},\qquad K_0=100,\qquad R=10^{200}.
\]
These choices specify the remaining constants in conditions 1 and 2.
Here $R$ is an auxiliary scale, not the initial index. Since
\[
 d=\frac{1+\log(11/10)}{5/4}<0.877,
 \qquad \frac1{10}<\sqrt{\frac{1-d-\Pi}{4}},
\]
the choice of $a_0$ satisfies (7.3). We verify explicit versions of the
estimates used above. All estimates below apply to arbitrary finite
histories satisfying the stated provisional invariants; none depends on
the finite terminal index.

\noindent\textbf{Condition 3: windows and inherited rows.}
Write $s=H^{-1}$ and suppose $s\geq R$. For $1\leq a\leq 11/10$, the
solution $x_a$ of $f_a(x_a)=H$ satisfies
\[
 s+\tfrac12s^{3/4}\leq x_a\leq s+2s^{3/4}.
\]
These inequalities follow by substituting the two endpoints in the
strictly decreasing function $f_a$. Differentiation gives
\[
 \frac{\partial x_a}{\partial a}
 =\frac{x_a^{3/4}(1+x_a^{-1})^2}
 {1+(5a/4)x_a^{-1/4}(1+x_a^{-1})^2},
 \qquad
 \left|s^{-3/4}\frac{\partial x_a}{\partial a}-1\right|
 \leq 10s^{-1/4}.
\]
Integration over $[1,11/10]$, followed by rounding the two endpoints,
therefore yields
\[
 \left|\frac{L(H)}{Ds^{3/4}}-1\right|\leq40s^{-1/4},
 \qquad .05s^{3/4}\leq L(H)\leq .2s^{3/4}.
\]
In particular, the window is nonempty. The same endpoint bounds imply
\[
 a_{\min}\geq(1-2s^{-1/4})H,\qquad
 a_{\max}\leq H-\tfrac14H^{5/4},\qquad
 a_{\max}-a_{\min}\leq H^{5/4}.
\]
Thus, for $H\leq W\leq16H^{4/5}$,
\[
 1\leq\lfloor W/H\rfloor\leq k\leq32s^{1/5}\leq L(H),
 \qquad
 cH^{5/4}\leq e\leq H(1+200s^{-1/20}).
\]
For the last inequality, use
$e<cH^{5/4}+a_{\max}+k(a_{\max}-a_{\min})$.
Since $200R^{-1/20}=2\cdot10^{-8}<\eta$, both the width lower bound
$a_{\min}\geq(1-\eta)H$ and the endpoint bound $e\leq(1+\eta)H$
hold. Principal gaps have their stated short-side orientation, since
$Bi^{-5/4}<1/i$. The endpoint children have short sides at most
$(2/3)(1+\eta)H<3H/4$, and the lower reserve gives their shape bound
with $C=16$. Consequently all the conclusions and numerical contraction
constants of Sections 2 and 3 apply.

\noindent\textbf{Condition 4: complete bins and main rows.}
The two discarded boundary-bin portions have total length at most
$4s^{13/20}$. Hence
\[
 L'\geq Ds^{3/4}(1-1000s^{-1/10}),\qquad
 1\leq L/L'\leq1+1000s^{-1/10}<100/99.
\]
In particular, every window contains complete bins.
For a main row with initial index $z\geq R$, the following complete bin
ends before $z+3z^{13/20}$. Its main positions have total width at least
$.02z^{-7/20}>z^{-1/2}$. Thus the greedy selection stops within that
bin. Using $.1z^{-1/2}\leq w\leq z^{-1/2}$ and the greedy remainder
bound gives
\[
 .05\sqrt z\leq k\leq2\sqrt z,\qquad
 z\leq i_j\leq z+3z^{13/20}.
\]
For these indices, and $1\leq U\leq1.1$,
\[
 \left|i_j^{5/4}\left(\frac1{z+1}+Uz^{-5/4}
                         -\frac1{i_j+1}\right)-U\right|
 \leq20z^{-1/10}.
\]
Thus $K_0=100$ suffices in (5.4). The height interval is nonempty since
$200z^{-1/10}<D$, and $H\leq3/(2z)$.
The endpoint estimates in Section 5 follow explicitly from
\[
 c(3/2)^{5/4}<1/4,\qquad
 \frac1z+\frac14z^{-5/4}
 \leq\frac1{z+1}+z^{-5/4}.
\]
All inherited-row spans are at most $3z^{3/4}$ when their first assigned
index is $z$; the same bound includes main rows.

\noindent\textbf{Condition 4 (continued): endpoint budgets and the area bootstrap.}
For main rows starting in $[T,2T)$, their distinct assigned indices lie
in $[T,3T]$, and each row uses at least $.05\sqrt T$ indices. There
are at most $60\sqrt T$ such rows. Summing their heights over dyadic
intervals gives
\[
 \sum_{\rm main\ rows}H\leq1000m^{-1/2}.
\]
The unique source indices of principal-gap rows give
\[
 \sum_{\rm principal\text{-}gap\ rows}H
 \leq B\sum_{i\geq m}i^{-5/4}\leq6m^{-1/4}.
\]
Both endpoint accounts of Section 7 are consequently less than
\[
 125(1000m^{-1/2}+6m^{-1/4})
 \leq 2\cdot10^5m^{-1/4}<\Pi/2.
\]
Integral comparison in (6.2), including integer endpoint rounding,
gives the explicit bound
\[
 D_P(t)\leq dt+4t^{9/5}.
\]
Also $t_N\leq N^{-1}(1+2N^{-1/4})$, and the row-span bound implies
that the total reserved area is at most $4N^{-5/4}$. Thus
\[
 A_N\geq\frac{1-d-\Pi-20N^{-1/4}}{N}
       >\frac{.12}{N}.
\]
With aspect ratio at most two, the shorter side is therefore at least
$\sqrt{.06}\,N^{-1/2}>a_0N^{-1/2}$. A strip of height at most
$3/(2N)$ satisfies $H\leq w/2$, and Lemma 5.2 preserves the aspect
ratio. This supplies
the strict improvement required in the first-failure argument.
All created row heights are at most $2/m$: main heights are at most
$3/(2m)$, principal-gap heights are at most $Bm^{-5/4}$, and endpoint
heights decrease at each split. Hence the preceding small-height
estimates apply throughout that induction.

\noindent\textbf{Condition 3 (continued): source supports and causal separation.}
Let $S=N^{4/5}$. Directly from (8.4) and
$|Nf_a(N)-1|\leq2N^{-1/4}$, the source support lies in
\[
 [S-2S^{3/4},\ B^{4/5}S+2S^{3/4}].
\]
At these scales a complete job window has diameter at most
$.2(2S)^{3/4}<S^{3/4}$; a boundary bin has length at most
$2S^{13/20}<S^{3/4}$. These bounds justify, respectively, the inner,
working, and comparison envelopes with coefficients $4,6,8$ used in
Section 10. Since $B^{4/5}<1.080$ and $8S^{-1/4}<.001$, those
envelopes fit strictly inside the stated fixed-proportion bands. An
additional intersecting window also fits in the next band, since its
diameter is less than $S^{3/4}$.

The endpoint parent-to-child activation ratio is at most
$(2/3)(1+\eta)\cdot1.01<.70$ (the parent activation index is at
most $1.01/H$); a principal-gap supplier is
created by time at most $2s^{4/5}<.70s$. Allowing the window and bin
rounding errors still gives creation before $.75s$ and complete
revelation before the first eligible index. At source scale $S$, the
bands therefore give supplier creation before $.825S$ and earliest
quota revelation after $.882S-10S^{3/4}>.85S$.
The analogous bin freeze at $.8s_b$ lies strictly between supplier
creation and first quota revelation. The integer floors are harmless,
since each remaining time margin exceeds one.
Main-row look-ahead also uses already revealed bins, since
$.9(z+3z^{13/20})<z$.
Finally
\[
 S-2S^{3/4}>.85S+3(.85S)^{3/4},
\]
so an actual contributing row has passed the nominal source freeze.
The initial-history cases follow as in Section 10: suppliers would
otherwise have to be created before the initial cursor, and their pool
is empty. All comparison indices are below $2N^{4/5}<N$, as required
by the least-obstruction argument. A contributing main-row member,
together with the main-row span bound, forces its reference into
$[S/2,2S]$, justifying the query restriction in Section 8.7.
These checks cover the causal and boundary conditions without
any dependence on the terminal index.

\noindent\textbf{Condition 5: the sharp mean estimates.}
For an eligible output gap, the window bounds above imply
\[
 \lambda_N(i,g)\leq
 (1+100N^{-1/4})\frac{N^{1/4}}{Di}
 \leq40N^{-11/20}.
\]
The weights $w_i$ satisfy the same final bound.
For a main row, division of the density estimate by
$i^{1/4}/(DN)$ gives at most
\[
 \frac{1+100N^{-1/4}}{1-2000z^{-1/10}}
 \leq1+10^6S^{-1/10},\qquad S/2\leq z\leq2S.
\]
Here $z\leq i$, so the additional factor $(z/i)^{5/4}$ can be
discarded. For an inherited input row with reciprocal height $s$,
where $S/2\leq s\leq2S$, the inverse-function count is bounded by
\[
 DN^{-5/4}s^2(1+20S^{-1/4})+2.
\]
Its additive rounding term has relative size at most $80S^{-7/16}$.
Since all eligible positions satisfy $i\geq s$, division of its
formal mean by $ks^{1/4}/(DN)$ is at most
\[
 \frac{(1+100N^{-1/4})
       (1+20S^{-1/4}+80S^{-7/16})}
      {1-1000s^{-1/10}}
 \leq1+10^6S^{-1/10}.
\]
For $S\geq R$, $10^6S^{-1/10}\leq10^{-14}<10^{-3}$.
This proves both fixed-margin comparisons (8.6) and (8.10) with
the actual weight $w_i$ in the manuscript. In the inherited comparison,
$\sum_iq_{r,i}i^{1/4}\geq k_rs^{1/4}$, so no limiting equality of
these two expressions is required.

Integral comparison on the envelope with coefficient six gives
\[
 \sum_{i\in\Omega_N}\frac{i^{1/4}}{DN}
 \leq .8+1000S^{-1/4}\leq .801.
\]
Deleting positions below $m$ only decreases this sum. Consequently
\[
 \mathbb EY\leq1.001\cdot.801=.801801<.81,
 \qquad Q_i\leq(100/99)\rho_I<1.
\]
These are the sharp mean and nonnegative-mass requirements of Section 8.

\noindent\textbf{Condition 5 (continued): explicit concentration constants.}
For one bin starting at $b$, the retained-window estimate gives mean
occupancy at most $.88\ell_b$. Every job quota is at most
$10^4b^{1/10}$. Lemma A.2, with excess $.02\ell_b$ and
$\ell_b\geq\tfrac12 b^{13/20}$, therefore gives
\[
 \mathbb P\{\hbox{bin occupancy}\geq .90\ell_b\}
 \leq\exp(-10^{-10}b^{11/20}).
\]
There are at most $3S$ working bins, all starting at least $S/2$.
For $S\geq R$, their union probability is at most
$\exp(-10^{-12}N^{11/25})$; explicitly,
$\log(3S)\leq10^{-11}S^{11/20}$ absorbs the counting factor.

The input-load bound on the complete comparison windows gives at most
$3S$ jobs and at most $3S$ working positions. Each job occupies at most
$10S^{1/10}$ complete bins. With $h_N=40N^{-11/20}$, the sum of
squared quota-mark oscillations is at most
\[
 3S(400S^{1/10}N^{-11/20})^2
       =480000N^{-7/50}.
\]
The corresponding permutation-exposure sum is at most
\[
 3S(80N^{-11/20})^2=19200N^{-3/10}.
\]
Lemma A.1 with excess $.01$ bounds these two probabilities by
$\exp(-10^{-12}N^{7/50})$ and
$\exp(-10^{-12}N^{3/10})$, respectively. Thus the source-only
bound holds with $\kappa_s=10^{-12}$ and prefactor three.

For an endpoint split, the same load estimate gives conditional mean
at most
\[
 20(1+100N^{-1/4})(1+N^{-1})H\leq21H.
\]
Indeed $eN\geq(3/2)N/(N+1)$ whenever either child can be eligible.
Moreover $e\leq3f_B(N)\leq6/N$ and the reserve imply
$H\leq(48/N)^{4/5}$. Summing the two child loads gives the
generous increment bound $10^4N^{-11/20}$. With mean budget $.0021$
and excess $.0079$, Lemma A.2 gives an endpoint decay constant at
least $10^{-10}$.

For an actual main query, the inverse-function count used above is
at most $2Ds_H^2N^{-5/4}$, including integer rounding, and
$s_H\leq4S$. Multiplication by $40N^{-11/20}$ gives the bound
$128N^{-1/5}$ in (8.19). The decay constant in (8.20), at excess
$.02$, is
\[
 \frac{(.02)^2}{256(.81+.02/3)}>10^{-6}.
\]
The actual-history and persistent-detection arguments of Sections 10
and 11 preserve these bounds. In particular it is safe to use the
single smaller constant
\[
                 \kappa_* =10^{-20}
\]
in all three terms of (11.8), keeping their existing prefactors.

\noindent\textbf{A single numerical threshold.}
For $m\geq10^{1000}$, every reciprocal row height is at least $m/2$,
and every source or comparison scale used above is at least
$m^{4/5}/4>R$. Thus all the estimates just verified apply
simultaneously. The endpoint-budget and area inequalities also hold
with strict margins at this value of $m$.
Finally, for $n\geq m$, the elementary inequality
$\log n\leq(100/7)n^{7/100}$ and
\[
 n^{7/100}\geq10^{70}>
 \frac{200}{7\kappa_*}
\]
give $\kappa_*n^{7/50}\geq2\log n$. Consequently \eqref{eq:12.8} is at most
\[
 6\sum_{n=m}^{\infty}e^{-\kappa_*n^{7/50}}
 \leq6\sum_{n=m}^{\infty}n^{-2}
 \leq\frac6{m-1}<\frac12.
\]
The finite-horizon construction therefore has positive probability
of success for every terminal index. Its termination requires no
additional numerical threshold. Section \ref{sec:13} supplies the infinite
packing and the equal-area identity without further enlargement of
$m$. Thus $m_0=10^{1000}$ satisfies all five requirements and is a sufficient
threshold in Theorem~\ref{thm:1.1}.

The exponent $1000$ is chosen solely for convenience. The proof above
does not claim optimality, nor does it address the case $m=1$.

\newpage
\begin{appendices}

\section{Two probability inequalities}\label{sec:8}

The random sample space used below is a countable product of uniform intervals and finite permutation spaces, with their usual Borel sigma algebras. The finite adaptive histories are measurable functions of this space. All conditional expectations and tail estimates refer to these actual measurable random variables.

We state the finite forms used below to make the concentration arguments self-contained. Every application is first restricted to a finite horizon; no direct application to an infinite random process is needed.

\begin{lemma}[Bounded differences]
\label{thm:8.1}
Let \(Y\) be a function of independent random variables \(V_1,\ldots,V_J\). Suppose that changing \(V_j\) alone changes \(Y\) by at most \(a_j\). Then

\[
\mathbb P\{Y-\mathbb EY\ge u\}
\le\exp\left(-\frac{2u^2}{\sum_j a_j^2}\right).
\tag{A.1}\label{eq:8.1}
\]

The same bound holds for \(\sum_jD_j\) if, conditional on the past, each martingale difference \(D_j\) takes values in an interval of length at most \(a_j\). Here \(u>0\). If \(\sum_j a_j^2=0\), the relevant random variable is constant and the upper-tail probability is zero.

\end{lemma}

\begin{proof}
If a centered random variable \(X\) takes values in an interval of length \(a\), then

\[
\mathbb Ee^{\lambda X}\le e^{\lambda^2a^2/8}.
\]

To see this, set \(\psi(\lambda)=\log\mathbb Ee^{\lambda X}\). Its second derivative is the variance under the exponentially tilted distribution. Every probability distribution supported on an interval of length \(a\) has variance at most \(a^2/4\). Integrating twice and using \(\psi(0)=\psi'(0)=0\) proves the displayed estimate. Applying its conditional version successively, followed by Markov's inequality with \(\lambda=4u/\sum a_j^2\), gives \eqref{eq:8.1}. For a function of independent coordinates, the Doob martingale obtained by revealing the coordinates one at a time has conditional oscillation at most \(a_j\) at step \(j\). It therefore falls under the martingale version.

\end{proof}

\begin{lemma}[Predictable mean budget]
\label{thm:8.2}
Let \((\mathcal F_j)_{j=0}^J\) be a filtration, and let \(X_j\), for \(1\le j\le J\), be nonnegative \(\mathcal F_j\)-measurable random variables revealed sequentially. Suppose that

\[
0\le X_j\le b,\qquad
\mu_j=\mathbb E[X_j\mid\mathcal F_{j-1}],\qquad
\sum_j\mu_j\le M,\qquad b,M\ge0
\tag{A.2}\label{eq:8.2}
\]

almost surely, where \(b\) and \(M\) are deterministic. For \(u>0\),

\[
\mathbb P\left\{\sum_j(X_j-\mu_j)\ge u\right\}
\le\exp\left[-\frac{u^2}{2(bM+bu/3)}\right].
\tag{A.3}\label{eq:8.3}
\]

Consequently the same bound holds for \(\{\sum_j X_j\ge M+u\}\). More generally, if a possibly random comparison weight \(W\) satisfies \(\sum_j\mu_j\le W\) pointwise, it holds for \(\{\sum_jX_j-W\ge u\}\). The deterministic \(M\) bounds the predictable mean budget; \(W\) need not be deterministic. When \(b=0\) or \(M=0\), these events have probability zero and the displayed fraction is not used.

Whether a step is performed may be determined by the preceding information. After stopping, the remaining increments are set to zero. Later we use zero-based program indexing: then \(X_j\) is \(\mathcal F_{j+1}\)-measurable and its mean is \(\mathbb E[X_j\mid\mathcal F_j]\), for \(0\le j<t\). This is just a shift of the indices in the finite statement.

\end{lemma}

\begin{proof}
First suppose \(b>0\) and \(M>0\), and put \(D_j=X_j-\mu_j\). Convexity on \([0,b]\) gives the chord bound
\[
 e^{\lambda x}\le1+\frac{x}{b}(e^{\lambda b}-1).
\]
Taking a conditional expectation, using \(1+y\le e^y\), and centering gives
\[
 \mathbb E[e^{\lambda D_j}\mid\mathcal F_{j-1}]
 \le\exp\left[\frac{\mu_j}{b}
 (e^{\lambda b}-1-\lambda b)\right].
\]
For \(0\le x<3\), the exponential series and \(k!\ge2\cdot3^{k-2}\) for \(k\ge2\) imply
\[
 e^x-1-x\le\frac{x^2}{2(1-x/3)}.
\]
Therefore, for \(0<\lambda<3/b\),
\[
\mathbb E[e^{\lambda D_j}\mid\mathcal F_{j-1}]
\le\exp\left(\frac{\lambda^2b\mu_j}{2(1-\lambda b/3)}\right).
\tag{A.4}\label{eq:8.4}
\]
The process
\[
\exp\left(\lambda\sum_{j\le k}D_j
-\frac{\lambda^2b}{2(1-\lambda b/3)}\sum_{j\le k}\mu_j\right)
\]
is consequently a nonnegative supermartingale with initial expectation one. On the event \(\sum D_j\ge u\), its terminal value is at least
\(\exp(\lambda u-\lambda^2bM/[2(1-\lambda b/3)])\).
Markov's inequality and
\[
 \lambda=\frac{u}{bM+bu/3},\qquad 0<\lambda b<3,
\]
give \eqref{eq:8.3}. The two stated consequences follow from
\(\sum X_j-W\le\sum(X_j-\mu_j)\) when \(W\ge\sum\mu_j\).
If \(b=0\), each increment is zero. If \(M=0\), nonnegativity implies \(\mu_j=0\) almost surely and hence \(\mathbb E X_j=0\), so again each increment is zero almost surely. A finite intersection of these full-probability statements handles all increments together. Stopped steps contribute zero and satisfy the same bounds.

\end{proof}

\section*{Acknowledgements}

The author thanks Jianan Shao for suggesting the problem that motivated this work. The author also gratefully acknowledges the outstanding assistance of OpenAI's GPT-6 in developing and refining the arguments, working through quantitative estimates, and improving the exposition.

\end{appendices}

\phantomsection
\addcontentsline{toc}{section}{References}
\bibliographystyle{unsrt}
\bibliography{references}

@article{meir1968packing,
  title={On packing of squares and cubes},
  author={Meir, Aram and Moser, Leo},
  journal={Journal of combinatorial theory},
  volume={5},
  number={2},
  pages={126--134},
  year={1968},
  publisher={Elsevier}
}

@article{tao2024perfectly,
  title={Perfectly packing a square by squares of nearly harmonic sidelength},
  author={Tao, Terence},
  journal={Discrete \& Computational Geometry},
  volume={71},
  number={4},
  pages={1178--1189},
  year={2024},
  publisher={Springer}
}

@article{zhu2022packing,
  title={Packing $1.35\cdot 10^{11}$ rectangles into a unit square},
  author={Zhu, Mingliang and Jo{\'o}s, Antal},
  journal={arXiv preprint arXiv:2211.10356},
  year={2022}
}

@article{kislovskiy2026slack,
  title={Slack-Pack Algorithm for Meir--Moser Packing Problem},
  author={Kislovskiy, Alexey D and Lerner, E Yu and Senkevich, Igor A},
  journal={Lobachevskii Journal of Mathematics},
  volume={47},
  number={5},
  pages={2346--2363},
  year={2026},
  publisher={Springer}
}

\end{document}